%% file: TWC_R2.tex
\documentclass[10pt]{IEEEtran}
\ifCLASSINFOpdf
\else
\fi
\usepackage[ colorlinks = true,
linkcolor = black,
urlcolor  = black,
citecolor = red,
anchorcolor = green,]{hyperref}
\input{preamble}

\usepackage{cite}
\usepackage{mathrsfs}
\usepackage{amsmath}
\usepackage{amsmath,amssymb,amsfonts}
\usepackage{amssymb}
\usepackage{makecell}

\usepackage{tikz}
\usepackage{textcomp}
\usepackage{algorithm}
\usepackage{xcolor}
\usepackage{amsthm,algorithmic}
\newtheorem{thm}{\bf Theorem}
\usepackage{setspace} 
\usepackage{lettrine}
\usepackage{stfloats}

\begin{document}

\title{Collaborative Secret and Covert Communications for  Multi-User Multi-Antenna Uplink UAV Systems: Design and Optimization}
    \author{Jinpeng Xu, Lin Bai, Xin Xie, and Lin Zhou
    \thanks{Jinpeng Xu, Lin Bai, and Lin Zhou are with the School of Cyber Science and Technology, Beihang University, Beijing, China, 100191. E-mails: \{xujinpeng, l.bai, lzhou\}@buaa.edu.cn \textit{(Corresponding author: Lin Zhou).}}
    \thanks{Xin Xie is with the School of Electronics and Information Engineering, Beihang
University, Beijing, China, 100191. E-mail: xiexincqbz@buaa.edu.cn.}
\thanks{This paper was partially presented at the 2024 IEEE ICCC workshop~\cite{xu2024secret}.}
}

\maketitle
\begin{abstract}
Motivated by diverse secure requirements of multi-user in UAV systems, we propose a collaborative secret and covert transmission method for multi-antenna ground users to unmanned aerial vehicle (UAV) communications. Specifically, based on the power domain non-orthogonal multiple access (NOMA), two ground users with distinct security requirements, named Bob and Carlo, superimpose their signals and transmit the combined signal to the UAV named Alice. An adversary Willie attempts to simultaneously eavesdrop Bob's confidential message and detect whether Carlo is transmitting or not. We derive close-form expressions of the secrecy connection probability (SCP) and the covert connection probability (CCP) to evaluate the link reliability for wiretap and covert transmissions, respectively. Furthermore, we bound the secrecy outage probability (SOP) from Bob to Alice and the detection error probability (DEP) of Willie to evaluate the link security for wiretap and covert transmissions, respectively. To characterize the theoretical benchmark of the above model, we formulate a weighted multi-objective optimization problem to maximize the average of secret and covert transmission rates subject to constraints SOP, DEP, the beamformers of Bob and Carlo, and UAV trajectory parameters. To solve the optimization problem, we propose an iterative optimization algorithm using successive convex approximation and block coordinate descent (SCA-BCD) methods. Our results reveal the influence of design parameters of the system on the wiretap and covert rates, analytically and numerically. In summary, our study fills the gaps in collaborative secret and covert transmission for multi-user multi-antenna uplink UAV communications and provides insights to construct such systems.
\end{abstract}


\section{Introduction}
%\IEEEPARstart{U}{nmanned}
Unmanned aerial vehicle (UAV) communication has attracted significant interest in recent years, primarily due to its vital role in constructing space-air-ground integrated networks (SAGIN)~\cite{mozaffari2019tutorial}. {However, UAVs typically have a certain flight altitude, which leads to the absence of physical barriers in UAV communications. Consequently, UAV communications often occur with a high probability of line-of-sight (LoS) transmission, and stronger LoS signals make confidential communications more susceptible to detection and interception by adversaries~\cite{XJP_Model}. Therefore, compared with land-based wireless communications, the high probability of LoS signal transmissions in UAV communications leads to severe security risks~\cite{adil2022systematic}, where the security risk includes both the confidentiality of communication content and the protection of communication behavior.} Specifically, the above characteristic presents a challenge for UAV secure communications to satisfy robust security communication requirements of next-generation networks~\cite{wei2022toward}. However, traditional encryption-based algorithms are vulnerable when adversaries have unlimited computational ability \cite{Survey_security_2016}. Fortunately, Shannon's pioneering theory of physical layer security (PLS) provides methods to achieve perfect security in the above case~\cite{Shannon1949}. Thus, PLS, along with traditional encryption technologies, is indispensable for establishing UAV secure wireless communication networks.

In the past few decades, to ensure communication security in PLS, extensive works~\cite{chen2016survey} concentrated on the wiretap channel theorem pioneered by Wyner~\cite{wyner1975wire}. Wiretap channels enable secure communications by using the characteristics of wireless channels, e.g., fading, noise, and interference~\cite{csiszar1978broadcast, trappe2015challenges}. Specifically, several channel coding methods have been proposed to simultaneously guarantee communication security and reliability, e.g., Wyner codes~\cite{wyner1975wire}, Polar codes~\cite{mahdavifar2011achieving}, and low-density parity check codes (LDPC)~\cite{klinc2011ldpc}. {The codes for wiretap channels safeguard the content of sensitive messages from being eavesdropped by malicious adversaries, which is insufficient for certain security-sensitive command and control systems where successful detection of communication occurrence between legitimate users is highly undesirable.}

In order to protect the communication occurrences of legitimate users from being detected, a promising new technology in PLS known as covert communication, or the low probability of detection communication, has been proposed~\cite{bash2013limits}. The seminal papers by Bash \emph{et al.}~\cite{bash2013limits}, Wang \emph{et al.}~\cite{wang2016fundamental}, and Bloch~\cite{bloch2016covert} demonstrated the now well-known square root law for covert communication over the additive white Gaussian noise channel, the discrete memoryless channel, and the binary symmetric channels. In particular, it was shown that $O(\sqrt{n})$ bits can be covertly and reliably transmitted over the above channels when $n$ channel uses are used. In order to achieve covert communication, several studies proposed low-complexity covert communication coding schemes~\cite{zhang2019covert,kadampot2020multilevel}.  These approaches aim to enhance the practicality and efficiency of covert communication systems, ensuring that the transmitted signals remain undetectable by adversaries. Different from the wiretap channel studies, covert communication safeguards the communication behaviors of legitimate users from being detected by any adversary~\cite{chen2023covert}. Thus, covert communication provides additional degrees of freedom for robust PLS~\cite{bai2023covert}.

Both wiretap channel studies and covert communication are fundamentals for PLS. Based on these theories, secrecy and covert communication have been studied comprehensively in various UAV scenarios, e.g.,~\cite{wang2022physical,chen2023covert}. In the remaining part of this section, we recall existing studies on secret and covert UAV communications and clarify our main contributions beyond these studies.

\subsection{Related Works}
The wiretap channel in UAV communications has been extensively investigated~\cite{wu2020energy,sun2019secure,liu2022ris,duo2020energy,li2021secrecy,zhang2019securing,li2021secure,wang2023secrecy}, including both downlink and uplink transmissions with various antenna setups and artificial noise (AN) technology. For UAV downlink secrecy communications, Wu \emph{et al.}~\cite{wu2020energy} proposed a transmission method to enhance the secrecy performance with the impact of UAV jitter, and optimized beamformers for confidential and AN signals to minimize the transmission power. With the directional modulation technology, Sun \emph{et al.}~\cite{sun2019secure} studied the security, reliability, and energy efficiency of downlink millimeter-wave (mmWave) simultaneous wireless information and power transfer UAV networks using non-orthogonal multiple access (NOMA) and orthogonal multiple access schemes. Liu \emph{et al.}~\cite{liu2022ris} studied the reconfigurable intelligent surface-mounted UAV downlink communication and optimized the secrecy rate using the successive convex approximation (SCA) and semidefinite relaxation (SDR) techniques. For UAV uplink secure communications, Duo \emph{et al.}~\cite{duo2020energy} studied a full-duplex secrecy communication scheme and optimized the UAV's trajectory and jamming power to maximize energy efficiency. Li \emph{et al.}~\cite{li2021secrecy} proposed an analogous scheme in UAV wireless sensor networks and maximized energy efficiency by optimizing the transmit power, jamming power, and UAV's trajectory.

In order to design an integrated security scheme for UAV downlink and uplink communications, Zhang \emph{et al.}~\cite{zhang2019securing} optimized the UAV trajectories and the transmit power to maximize secrecy rates for both UAV-to-ground and ground-to-UAV communications in the presence of potential eavesdroppers. Li \emph{et al.}~\cite{li2021secure} studied the secrecy performance of a UAV-to-vehicle communication system, and evaluated the secrecy performance for both downlink and uplink transmission. Wang \emph{et al.}~\cite{wang2023secrecy} maximized the secrecy rate and energy efficiency in UAV-enabled two-way relay systems by optimizing the UAV velocity, trajectory, and power allocation. All the aforementioned studies provided insight and guidance for designing content security requirements for UAV systems. 

To protect the communication behavior of legitimate users, covert UAV communications have been studied extensively~\cite{zhou2021three,yan2021optimal,Zhou_TrajPower,hu2019optimal,zhang2020optimized,jiang2021resource,Su_NOMA_downlink_uav_covert,li2021md,zhou2021uav,zhang2022uav,chen2021uav_tvt,chen2021uav_tcom,wang2019secrecy}. In the scenario with one legitimate receiver, Zhou \emph{et al.}~\cite{zhou2021three} and Yan \emph{et al.}~\cite{yan2021optimal} optimized the placement and the transmit power of the UAV to improve the covert communication performance. Considering the adversary's location and noise uncertainties,  Zhou \emph{et al.}~\cite{Zhou_TrajPower} studied the UAV downlink covert communication and proposed a SCA algorithm to optimize the trajectory and the transmit power of the UAV. In UAV mmWave covert communications, Hu \emph{et al.}~\cite{hu2019optimal} used an antenna array with beam sweeping to detect the UAV transmission, and optimized the detection accuracy. Zhang \emph{et al.}~\cite{zhang2020optimized} proposed a beam-sweeping scheme to optimize the UAV covert rate when the ground transmitter is ignorant of the precise location of the UAV and the ground warden. Since UAV systems typically involve multiple users, Jiang \emph{et al.}~\cite{jiang2021resource} studied the UAV downlink multi-user covert communication, and optimized the UAV's transmit power and trajectory to maximize the covert rate using block coordinate descent (BCD) algorithm. Su \emph{et al.}~\cite{Su_NOMA_downlink_uav_covert} studied NOMA downlink covert communications with multi-users. 

To improve covert performance, AN-aided covert communications are widely explored in UAV multi-user communications. In particular, Li \emph{et al.}~\cite{li2021md} proposed an AN-aided covert transmission scheme cognitive radio networks and obtained a near-optimal covert rate. Zhou \emph{et al.}~\cite{zhou2021uav} and Zhang \emph{et al.}~\cite{zhang2022uav} employed the full-duplex architecture on UAV platforms to transmit AN against the warden's detection. In multiple warden scenarios, Chen \emph{et al.}~\cite{chen2021uav_tvt} used a multi-antenna jamming scheme to maximize the covert rate.  Considering that the UAV acts as a warden in multi-relay scenarios, Chen \emph{et al.}~\cite{chen2021uav_tcom} further proposed a UAV-relayed covert communication scheme with finite blocklength to maximize effectively covert transmission bits. Finally, Wang \emph{et al.}~\cite{wang2019secrecy} optimized the secrecy and covert throughput in multi-hop relay communication to obtain the optimal hops for transmission efficiency. 

In summary, although secret and covert communications have been extensively studied in various UAV communication scenarios, the study of multi-user multi-antenna collaborative secret and covert communications in UAV uplink systems has never been explored. Specifically, most studies on UAV secure communications suffer from one of the following limitations. Firstly, most papers focused on single-antenna scenarios~\cite{sun2019secure,duo2020energy,li2021secrecy,zhang2019securing,li2021secure,wang2023secrecy,zhou2021three,yan2021optimal,Zhou_TrajPower,jiang2021resource,li2021md,zhou2021uav,zhang2022uav,chen2021uav_tcom,wang2019secrecy}, which can not utilize spatial resources and beamforming to improve secure performance. Secondly, all existing studies for UAV secure communications focused on either secret communications using the wiretap channel theory~\cite{wu2020energy,sun2019secure,liu2022ris,duo2020energy,li2021secrecy,zhang2019securing,li2021secure,wang2023secrecy} or covert communications~\cite{zhou2021three,yan2021optimal,Zhou_TrajPower,jiang2021resource,chen2021uav_tvt,li2021md,zhou2021uav,zhang2022uav,chen2021uav_tcom,wang2019secrecy,hu2019optimal,zhang2020optimized}. But a single security measure cannot satisfy diverse security requirements of individual users. Thirdly, although several studies employed the AN technology to improve the secrecy or covert rate~\cite{duo2020energy,li2021secrecy,li2021md,zhou2021uav,zhang2022uav}, the integration of AN technology on UAV platforms is challenging due to its high power consumption and the limited battery power of UAV platforms. Finally, in existing multi-user UAV communications~\cite{li2021secrecy,zhou2021uav,wang2019secrecy}, ground users are assumed to transmit independently. However, it is challenging and insufficient to satisfy secure quality of service (QoS) without the cooperation of multiple users.

Since UAV communication often involves multiple users, ground users typically have various security communication requirements. For instance, some users prioritize content security, while others prioritize communication occurrence security. Although Forouzesh \emph{et al.}~\cite{forouzesh2020joint} studied the secret and covert communications in a downlink ground system. Note that there are many differences between~\cite{forouzesh2020joint} and our study. Specifically, i) Ref.~\cite{forouzesh2020joint} considered a downlink broadcast channel scenario in a ground-based environment, whereas our study focuses on an uplink multi-access channel scenario in UAV communications, ii) Ref.~\cite{forouzesh2020joint} treated eavesdropping and detection as distinct adversarial actions performed by different adversaries, making the theoretical analysis relatively straightforward. In contrast, our study considers eavesdropping and detection as actions carried out by the same adversary, which escalates
the complexity of detection and surveillance. Therefore, the results in~\cite{forouzesh2020joint}  are not suitable for uplink UAV communications since i) UAV communications have different channel conditions from ground communication, ii) the uplink communication strategy is vastly different from downlink, and iii) there is an essential difference when the eavesdropper and the warden are treated as an integration adversary versus a separation adversary. Thus, the study on the collaborative secret and covert communications for uplink UAV systems is still in its infancy, and the existing studies can not satisfy the diverse security requirements of multiple users. Therefore, to satisfy ground users' diverse security requirements, it is indispensable to study collaborative secret and covert communications for UAV systems. Subsequently, the key challenges to be addressed include designing an effective multi-user collaborative secure transmission strategy, balancing the trade-off between secret and covert rates, and optimizing resource allocation to enhance the QoS of ground users.

\subsection{Main Contributions}
To fill the aforementioned research gaps, we propose a NOMA-based multi-user corroborative secret and covert uplink transmission scheme for UAV communications over a quasi-static fading channel. Specifically, we study the secure performance of the proposed scheme and optimize system parameters and the UAV's trajectory to enhance the secure QoS. Our detailed contributions are summarized as follows.

\begin{enumerate}
\item To our best knowledge, this is the first paper that studies the collaborative secret and covert uplink communications for multi-user multi-antenna UAV systems. From this perspective, our results provide new insights and theoretical benchmarks for the design of secure UAV systems where users have diverse security requirements.

\item  We propose a strategy that collaboratively enables secret and covert uplink transmission for different users to the UAV over a quasi-static fading channel using the NOMA technology. Specifically, we derive close form expressions for the secrecy connection probability (SCP) and covert connection probability (CCP) to evaluate secret and covert connection performance, respectively. Furthermore, we derive close form expressions for the secrecy outage probability (SOP) of the secret user and the detection error probability (DEP) of the adversary to evaluate the link security for wiretap and covert users, respectively.

\item To capture the theoretical benchmark of our proposed scheme, we formulate a multi-objective optimization problem to maximize the weighted sum of the average secret rate and average covert rate subject to SOP, DEP, beamformers of two ground users, and the trajectory parameters of the UAV. {To solve the optimization problem, we propose an SCA-BCD algorithm, which provides feasible solutions for both the secret and covert beamformers, and the UAV's trajectory and flight parameters.} 

\item Using the proposed SCA-BCD algorithm, we uncover the trade-off between the secret rate and the covert rate and obtain the near Pareto front of the two optimization objectives. Furthermore, we reveal the influence of design parameters on the weighted sum of average secret rate and average covert rate, analytically and numerically. In particular, we verify that either the secret rate or the covert rate can dominate the sum of the rate under different conditions.
\end{enumerate}

\subsection{Organization for the Rest of the Paper}
The rest of this paper is organized as follows. In Section \ref{sec:model}, we present the system model, specify the channel model, describe the transmission scheme, and formulate optimization problems. In Section \ref{sec:main_results}, we present our main results including closed form expressions for optimization constraints and solutions to our proposed optimization problems. Subsequently, in Section \ref{sec:simulation}, we present numerical examples to illustrate our theoretical findings. Finally, we conclude the paper and discuss future research directions in Section \ref{sec:conclusion}.

\subsection*{Notation}
Random variables and their realizations are denoted by upper case variables (e.g.,  $X$) and lower case variables (e.g.,  $x$), respectively. All sets are denoted in the calligraphic font (e.g.,  $\mathcal{X}$). Let $X^{n}:=(X_1,\ldots,X_n)$ be a random vector of length $n$. Let $\mathcal{CN}(\bf 0, \bf\Sigma )$ denote a circularly symmetric complex Gaussian vector with zero mean and covariance $\bf\Sigma$. Given any integers $[m,n]$ and a $m\times n$ matrix $\bA$, let $\bA^T$, $\bA^H$,  $\mathrm{Tr}(\bA)$, and  $\mathrm{det}(\bA)$ denote its transpose, conjugate transpose, the trace of $\bA$, and the determinant of $\bA$, respectively. Analogously, given any integer $n$ and a length-$n$ vector $\ba$, let $\ba_{\calX}=\{\ba[i]\}_{i\in\calX}$ represents a vector consisting of elements in $\ba$ with indices belonging to a set $\calX\subseteq[1:n]$. 

\section{System Model and Transmission Scheme}
\label{sec:model}

As shown in Fig. \ref{fig:system_model}, two ground legitimate users, Bob and Carlo, aim to transmit signals to the legitimate user Alice, who is a fixed-wing UAV\footnote{Note that rotorcraft UAV communication can be regarded as a special case of fixed-wing UAV communication within a specific time slot. In order to consider a more general UAV communication scenario and provide new insights for future studies on covert UAV communications involving multiple ground covert users, we extend the rotorcraft UAV presented in the conference paper~\cite{xu2024secret} to a fixed-wing UAV in this study.}. Simultaneously, a passive adversary named Willie attempts to wiretap the confidential message from Bob and detect the communication occurrence of Carlo. Note that Bob and Carlo require secrecy and covertness of their messages, respectively, and we assume that Bob's and Carlo's codebooks are different~\cite{liang2008multiple}, with Willie having access to Bob's codebook but not to Carlo's.  Assume that Bob and Carlo are equipped with $N_\rmb$ and $N_\rmc$ antennas for beamforming transmission, respectively. {It is common in practical UAV systems to assume that the UAV is equipped with a single antenna due to size, weight, and power constraints. Consistent with previous studies~\cite{chen2021uav_tcom, chen2021uav_tvt,yan2021optimal}, we assume that both  Alice and Willie have a single antenna. In this setup, Willie does not require knowledge of the beamforming vectors and performs power detection for eavesdropping~\cite{ma2021robust}} 
\begin{figure}[bt]
\centering
\includegraphics[width=0.8\linewidth]{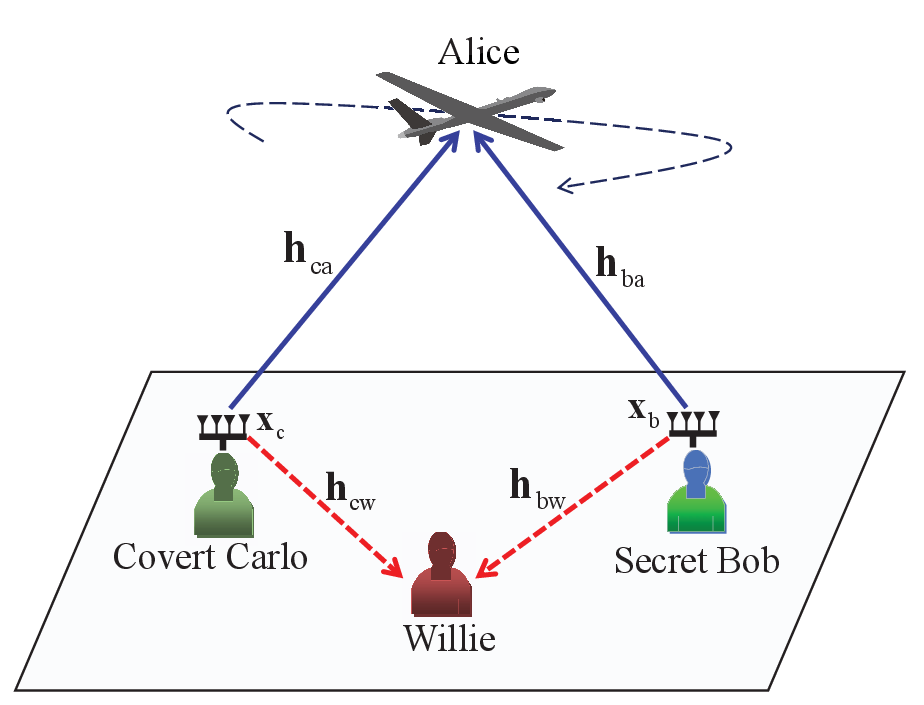}
\caption{System model of collaborative secret and covert uplink transmission for a multi-user multi-antenna UAV system.}
\label{fig:system_model} 
\end{figure}

{We denote the horizontal locations of Bob, Carlo, and Willie as ${\bL}_\rmb \buildrel \Delta \over = {\left[ {x_\rmb,y_\rmb} \right]^T} \in \mathbb R ^{2 \times 1}$, ${\bL}_\rmc \buildrel \Delta \over = {\left[ {x_\rmc,y_\rmc} \right]^T} \in \mathbb R ^{2 \times 1}$, and ${\bL}_\rmw \buildrel \Delta \over = {\left[ {x_\rmw,y_\rmw} \right]^T} \in\mathbb R ^{2 \times 1}$, respectively. Furthermore, consistent with previous studies~\cite{zhou2021uav,chen2021uav_tcom}, we assume that the locations of the ground users are available to Alice, since the accuracy position of the adversary can be obtained using the advanced computer vision technology~\cite{zhang2024earthgpt,zhang2024earthmarker}. Such an assumption allows for effective analysis of secure communication performance, especially when the adversary is either close to or far away from legitimate users.} Suppose that the UAV's flight period and altitude are $T$ and $H$, respectively. The UAV's horizontal coordinate represents as ${\bL}_\rma(t) \buildrel \Delta \over = {\left[ {x_\rma(t),y_\rma(t)} \right]^T} \in \mathbb R^{2 \times 1}$ where $0 \le t \le T$. Note that the continuous time $t$ implies an infinite number of velocity constraints, which makes the trajectory design intractable. To solve the problem, we divide the flight period $T$ into $N\in\bbN$ equal-length time slots. When $N$ is large, the period of each time slot $\delta_t = \frac{T}{N}$ is sufficiently small such that the UAV's location is approximately unchanged within each time slot. Hence, the UAV's trajectory can be denoted by ${\bL}_\rma[n]\buildrel \Delta \over = {\left[ {x_\rma[n],y_\rma[n]},H \right]^T} \in\mathbb R^{3 \times 1}$ where $n\in[N]$. The UAV's initial location and final location are denoted as ${\bL}_\mathrm{a}[0] \buildrel \Delta \over = {\bL}_\mathrm{0}$, ${\bL}_\mathrm{a}[N] \buildrel \Delta \over = {\bL}_{N}$, respectively. For each $n\in[N]$, the UAV's velocity and acceleration in $n$-th time slot are denoted as $\bv[n]$ and $\ba[n]$, respectively. The UAV's maximum velocity, minimum velocity, and maximum acceleration are denoted by $v_\mathrm{max}$, $v_\mathrm{min}$, and $a_\mathrm{max}$, respectively. Thus, during the flight period, the position and velocity of the UAV satisfy that for each $n\in[2:N]$,
\begin{align}
\|\bL_\rma[n]-\bL_\rma[n-1]\|^2 =&\bv[n]\delta_t+\frac{1}{2}\ba[n]\delta_t^2,\label{fly_step}\\
\|\bv_\rma[n]-\bv_\rma[n-1]\|^2=&\ba[n]\delta_t.\label{speed_step}
\end{align}

\subsection{Channel Model}
\subsubsection{Air to Ground Channel}
Consistent with~\cite{chen2021uav_tcom,Chen_UAV_backscatter_secure}, the channel coefficient between Bob/Carlo and Alice is assumed to follow the combination of a large-scale path-loss $\rma_\mathrm{\dagger a}[n]$ and a small-scale fading $\rmg_\mathrm{\dagger a}[n]$, where $\dagger\in\{\rmb, \rmc\}$ denotes either Bob or Carlo. Thus, we have
\begin{align}
{\bh_\mathrm{\dagger a}}[n] = {a_\mathrm{\dagger a}[n]}{\bg_\mathrm{\dagger a}[n]}.
\end{align}
The large-scale path-loss satisfies $\rma_\mathrm{\dagger a}[n] = \sqrt {{\lambda _0}\left(d_\mathrm{\dagger a}[n]\right)^{\xi _\mathrm{\dagger a}}}$,
where $\lambda _0$ is the path-loss reference at one meter, $d_\mathrm{\dagger a}[n]$ represents the distance between Alice and Bob, and $\xi_\mathrm{\dagger a}$ represents the air-ground path-loss exponents. Consistent with the previous study \cite{chen2021uav_tcom,jiang2021resource}, we set $\xi_\mathrm{\dagger a}=-2$.   The small-scale fading ${\bg_\mathrm{\dagger a}[n]}$ is assumed as Rician~\cite{WJX_URLLC} and satisfies
\begin{align}
{\bg_\mathrm{\dagger a}[n]} =
\sqrt {\frac{{{K_\mathrm{\dagger a}}}}{{{K_\mathrm{\dagger a}} + 1}}} {{\widehat \bg}_\mathrm{\dagger a}[n]} + \sqrt {\frac{1}{{{K_\mathrm{\dagger a}} + 1}}} {{\widetilde \bg}_\mathrm{\dagger a}[n]},
\end{align}  
where ${K_\mathrm{\dagger a}}$ is the Rician factor, ${\widehat \bg}_\mathrm{\dagger a}[n]$ and ${{\widehat \bg}_\mathrm{\dagger a}[n]}\! \!\sim\!{\mathcal{ CN}}(0,\bI)$ denotes the LoS channel component and the non-line-of-sight (NLoS) Rayleigh fading component, respectively. To simplify the analysis, we assume that ${K_\mathrm{\dagger a}}$ is constant during the UAV's flight period. This assumption is made because we do not focus on studying the impact of variations in the Rician factor on security performance, and this assumption is consistent with previous studies \cite{XJP_Model, chen2021uav_tcom}, and \cite{Chen_UAV_backscatter_secure}. 

\subsubsection{Ground Channel}
Consistent with~\cite{chen2021uav_tvt,he2024two}, the channel coefficient between Bob/Carlo and Willie is assumed to follow the combination of a large-scale path-loss ${a_\mathrm{\dagger w}[n]}$ and small-scale quasi-static Rayleigh fading ${\bg_\mathrm{\dagger w}[n]}$, and is denoted by
\begin{align}
{\bh_\mathrm{\dagger w}[n]} = {a_\mathrm{\dagger w}[n]}{\bg_\mathrm{\dagger w}[n]},
\end{align}
where  $a_\mathrm{\dagger w}[n]=\sqrt {\eta _0 {\lambda _0} \left(d_\mathrm{\dagger w}[n]\right)^{\xi _\mathrm{\dagger w}}}$, $\eta_0$ is the excessive path-loss coefficient, $d_\mathrm{\dagger w}[n]$ and $\xi_\mathrm{\dagger w}$ represent the distance and path-loss exponents between Bob and Willie, respectively. The small scale Rayleigh fading satisfies ${\bg_\mathrm{\dagger w}}[n] \sim {\cal{ CN}}(0,\bI)$. 

\subsection{Transmission Scheme}
We consider a transmission scheme that Carlo's covert signal hides within Bob's overt signal using the NOMA technology. Due to their different security requirements, Bob chooses a high transmit power to achieve 
a sufficient large secrecy throughput, while Carlo selects an appropriate transmission power to ensure covertness. Carlo achieves covert communication by leveraging Bob's transmission as a disguise. Consequently, Bob continuously transmits a sequence of $m$ symbols $\bx_\rmb[n]=(x^1_\rmb,...,x^m_\rmb)$ in $n$-th time slot, where for each $i\in [m]$, $\mathbb{E}[|x^i_\rmb|^2]=1$. Carlo occasionally transmits the another sequence of symbols $\bx_\rmc[n]=(x^1_\rmc,...,x^m_\rmc)$ in $n$-th time slot, where for each $i\in [m]$, $\mathbb{E}[|x^i_\rmc|^2]=1$. Willie attempts to eavesdrop the confidential message from Bob and detect whether Carlo is transmitting or not. 

Consistent with the previous studies~\cite{Tao_CovertNOMA, Hu_Inverfad}, we assume that Bob and Carlo are synchronized with facilitation by the UAV, thus the channel information state (CSI) of the legitimate users is known to them. The transmission includes two states, i.e., $\mathrm{H_0}$ and $\mathrm{H_1}$. At state $\mathrm{H_0}$, Bob transmits the secret signal sequence $\bx_\rmb[n]$ using a beamformer $\bw_\mathrm{b,0}$. At state $\mathrm{H_1}$, Bob transmits the secret signal $\bx_\rmb[n]$ using another beamformer $\bw_\mathrm{b,1}$, and Carlo transmits covert signals $\bx_\rmc[n]$ using the beamformer $\bw_\mathrm{c}$ to hide the covert signal $\bx_\rmc[n]$ inside the secret signal $\bx_\rmb[n]$. Let $Q_\mathrm{b}^\mathrm{max}$ and $Q_\mathrm{c}^\mathrm{max}$ denote the maximum transmit powers of Bob and Carlo, respectively. Thus, the beamformers satisfy $\max\{\|\bw_\mathrm{b,0}\|^2,\|\bw_\mathrm{b,1}\|^2\}\le Q_\mathrm{b}^\mathrm{max}$, and $\|\bw_\mathrm{c}\|^2\le Q_\mathrm{c}^\mathrm{max}$. {Using the power domain NOMA to superimpose the secret and covert signals, we can enhance the covertness while ensuring secrecy by hiding Carlo’s covert signal within Bob’s higher-power transmission, since Carlo’s signal can be considered as AN to decrease the SNR of Willie.  Therefore, the superposition of the secret and covert signals makes it more challenging for Willie to both eavesdrop on the secret signal and detect the covert transmission.}

For the legitimate receiver Alice, the secret signal $\bx_\rmb[n]$ is first decoded. Using successive interference cancellation (SIC) technology, the covert signal $\bx_\rmc[n]$ is decoded subsequently. On the adversary's side, Willie will receive either Bob's signal or the superimposed signal of both Bob and Carlo. Therefore, the received signals at Alice and Willie in the $n$-th time slot can be respectively expressed as
\begin{align}
\by[n] \!=\! \left\{\!\!\! \begin{array}{ll}
{\bh^H_\mathrm{ba}[n]}\bw_\mathrm{b,0}[n]\bx_\rmb[n] + {\bn_\mathrm{a}[n]},&\!\mathrm{H_0},\\
\\
{\bh^H_\mathrm{ba}[n]}\bw_\mathrm{b,1}[n]\bx_\rmb[n]\! + \! {\bh^H_\mathrm{ca}[n]}\bw_\mathrm{c}[n]\bx_\rmc[n]\! +  \!{\bn_\mathrm{a}[n]},&\!\mathrm{H_1},
\end{array} \right.
\end{align}
and
\begin{align}
\bz[n]\!=\! \left\{\!\!\! \begin{array}{ll}
{\bh^H_\mathrm{bw}[n]}\bw_\mathrm{b,0}[n]\bx_\rmb[n]  +   {\bn_\mathrm{w}[n]},&\!\mathrm{H_0},\\
\\
{\bh^H_\mathrm{bw}[n]}\bw_\mathrm{b,1}[n]\bx_\rmb[n] \!+\!    {\bh^H_\mathrm{cw}[n]}\bw_\mathrm{c}[n]\bx_\rmc[n] \! +  \! {\bn_\mathrm{w}[n]},&\!\mathrm{H_1},
\end{array} \right.
\end{align}
where ${\bn_\mathrm{a}}[n]\!\sim\! \mathcal{CN}(0,{\sigma}^2_\rma[n]) $ and ${\bn_\mathrm{w}}[n]\!\sim\! \mathcal{CN}(0,\sigma^2_\rmw[n]) $ are the additive noise at Alice and Willie in $n$-th time slot, respectively, and ${\sigma}^2_\rma[n]$ and $\sigma^2_\rmw[n]$ are the noise power in $n$-th time slote at Bob and Willie, respectively.

Using the proposed secure transmission scheme, the signal-to-noise ratio (SNR) values of secret signal $\bx_\rmb[n]$ at states $\mathrm{H_0}$ and $\mathrm{H_1}$ are respectively given by 
\begin{align}
{\mathrm{SNR_{ba,0}}[n]} &=    \frac{ \left| \bh^H_\mathrm{ba}[n]\bw_\mathrm{b,0}[n] \right|^2 }{{\sigma _\mathrm{a}^2[n]}},\\
{\mathrm{SNR_{ba,1}}[n]} &=    \frac{ \left| \bh^H_\mathrm{ba}[n]\bw_\mathrm{b,1}[n] \right|^2 }{\left| \bh^H_\mathrm{ca}[n]\bw_\mathrm{c}[n] \right|^2+\sigma _\mathrm{a}^2[n]}   \label{SINR_b_1}.
\end{align}
Using SIC, the SNR value of the covert signal $\bx_\rmc[n]$ is 
\begin{align}
{\mathrm{SNR_{ca}}[n]} =    \frac{ \left| \bh^H_\mathrm{ca}[n]\bw_\mathrm{c}[n] \right|^2 }{{\sigma _\mathrm{a}^2[n]}} .
\end{align}
For Willie, the SNR values of the secret signal $\bx_\rmb[n]$ at states $\mathrm{H_0}$ and $\mathrm{H_1}$ are respectively given by
\begin{align}
{\mathrm{SNR_{bw,0}}[n]} &= \frac{ \left| \bh_\mathrm{bw}^H[n]\bw_\mathrm{b,0}[n] \right|^2 }{{\sigma _\mathrm{w}^2[n]}} ,
\end{align}
\begin{align}
{\mathrm{SNR_{bw,1}}[n]} &= \frac{ \left| \bh_\mathrm{bw}^H[n]\bw_\mathrm{b,1}[n] \right|^2 }{\left| \bh_\mathrm{cw}^H[n]\bw_\mathrm{c}[n] \right|^2+\sigma _\mathrm{w}^2[n]}. 
\end{align}
Based on these SNR values in different cases, the secret and covert performance can be analyzed.

\subsection{Optimization Problem Formulation}
Let $R_\rmt$ and $R_\rmc$ represent the target transmission rate of Bob and the target covert rate of Carlo, respectively. Furthermore, let $R_\rms[n]$ and $R_\rmw[n]$ denote the target secret rate and the secret rate redundancy of the wiretap channel in $n$-th time slot involving legitimate users Alice and Bob and the malicious user Willie, respectively~\cite{csiszar1978broadcast}. Note that the secret rate redundancy $R_\rmw[n]:=R_\rmt-R_\rms[n]$ reflects the ability to secure the transmission against the wiretapping~\cite{wang2019secrecy}. Based on Shannon theory~\cite{Shannon1949}, Alice can decode the message from Bob correctly if the capacity $C_\mathrm{ba}$ for the channel from Bob to Alice is greater than the target transmission rate, i.e., $C_\mathrm{ba}>R_\rmt$, which is equivalent to $\mathrm{SNR_\mathrm{ba}}[n]\ge\gamma_\mathrm{ba}=2^{R_\rmt}-1$. Otherwise, a connection outage event occurs. According to Wyner's wiretap channel coding theory~\cite{wyner1975wire}, secret communication fails when the capacity $C_\mathrm{bw}[n]$ for the channel from the transmitter to Willie in $n$-th time slot is greater than the rate redundancy $R_\rmw[n]$, i.e., $C_\mathrm{bw}[n]\ge R_\rmw[n]$, which is equivalent to $\mathrm{SNR_\mathrm{bw}}[n]\!\ge\!\gamma_{\mathrm{bw}}[n]=2^{R_\rmw[n]}-1$, in this case, a secrecy outage event occurs.
%Note that $R_\rmw$ reflects the ability to ensure secrecy transmission against the wiretapping.

We assume that Bob and Carlo's target transmission rates $R_\rmt$ and $R_\rmc$ remain constant during the UAV's flight period. Therefore, the SCP between Bob and Alice in $n$-th time slot at states $\rmH_0$ and $\rmH_1$ can be respectively expressed as 
\begin{align}
{{\rmP^{\mathrm{{ba,0}}}_\mathrm{sc}}} [n]=1-\mathrm{Pr} \left\{{\mathrm{SNR_{ba,0}}}[n]\le 2^{R_\rmt}-1 \right\},
\end{align}
\begin{align}
{{\rmP^{\mathrm{{ba,1}}}_\mathrm{sc}}}[n] =1-\mathrm{Pr} \left\{{\mathrm{SNR_{ba,1}}}[n] \le 2^{R_\rmt}-1 \right\}.\label{psc1}
\end{align}
Analogously, the Bob's SOPs in $n$-th time slot at states $\rmH_0$ and $\rmH_1$ can be respectively given as 
\begin{align}
{{\rmP^{\mathrm{{bw,0}}}_\mathrm{so}}}[n] =\mathrm{Pr} \left\{{\mathrm{SNR_{bw,u0}}} [n] \ge 2^{R_\rmw[n]}-1 \right\},
\end{align}
\begin{align}
{{\rmP^{\mathrm{{bw,1}}}_\mathrm{so}}}[n] =\mathrm{Pr} \left\{{\mathrm{SNR_{bw,u1}}}[n] \ge 2^{R_\rmw[n]}-1 \right\}.
\end{align}
Furthermore, the CCP between Carlo and Alice in $n$-th time slot at state $\rmH_1$ can be given as 
\begin{align}
{\rmP^{\mathrm{{ca}}}_\mathrm{cc}}[n] =\!\mathrm{Pr}\! \left\{{\mathrm{SNR_{b,1}}[n]}  \!\ge \!2^{R_\rmt}\!-\!1, \mathrm{SNR_{\rmc}} [n]\!\ge\! 2^{R_\rmc}\!-\!1\right\}\!.
\end{align}
Finally, Willie's  DEP in the $n$-th time slot satisfies
\begin{align}
\mathrm{P_e}[n] = \mathrm{P_F}[n]+ \mathrm{P_M}[n],\label{dep}
\end{align}
where $\mathrm{P_F}[n]$ and $\mathrm{P_M}[n]$ are the false alarm (FA) and miss detection (MD) error probabilities in the $n$-th time slot, respectively, which will be analyzed in Eqs. \eqref{44} and \eqref{45}.

{Under the above definitions, we formulate the optimization problems under both states $\rmH_0$ and $\rmH_1$ subject to diverse security constraints and the UAV's flight parameters.
At each time slot $n\in[N]$, given legitimate users' beamformers $(\bw_\mathrm{b,1}[n],\bw_\mathrm{c}[n])$ and UAV's flight parameters $(\bv[n],\ba[n],{\bL}_\rma[n])$, the multi-objective optimization problem under $\rmH_1$ is can be given as follows
\begin{subequations}
\begin{align}  
{\rmP0}:\!\!\max_{\substack{\bw_\mathrm{b,1}[n],\!\bw_\mathrm{c}[n],\\\bv[n],\ba[n],{\bL}_\rma[n]}} \quad & \left[{\bf\Phi}_s,{\bf\Phi}_c\right]\label{obj_fun_P1}\\*
\mbox{s.t.}~~~~~~~~&{{\rmP^{\mathrm{{bw,1}}}_\mathrm{so}}}[n]\le \eta_\rms, \label{pso_1}\\*
& \mathrm{P_e}[n] \ge 1-\epsilon, \label{dep_constraint}\\*
& \|\bw_\mathrm{b,1}[n]\|^2\le Q_\mathrm{b}^\mathrm{max},\label{power_b_constraint_1}  \\    
&\|\bw_\mathrm{c}[n]\|^2\le Q_\mathrm{c}^\mathrm{max}, \label{power_c_constraint} \\ 
&v_\mathrm{min}\!\!\le\!\!\|\bv[n]\|^2 \!\!\le \!\!v_\mathrm{max},\label{velocity_constraint}\\ 
&\|\ba[n]\|^2\le a_\mathrm{max},\label{acc_constraint}\\
&\bL_\rma[0]= \bL_0,\bL_\rma[N]=\bL_N,\label{L_start_end}\\ 
&\bv[0]=v_0,\bv[n]=v_N,\label{v_start_end},\\*
&\eqref{fly_step},~\eqref{speed_step}\label{fly_speed_step}.
\end{align}
\end{subequations}
where ${\bf\Phi}_s=\frac{1}{N}\sum\limits_{n = 1}^N {{\rmP^{\mathrm{{ba,1}}}_\mathrm{sc}}[n]{R^1_\rms}[n]}$ and ${\bf\Phi}_c=\frac{1}{N}\sum\limits_{n = 1}^N {{\rmP^{\mathrm{{ca}}}_\mathrm{cc}}[n]{R_\rmc}}$ represent Bob's average secret and Carlo's average covert rate. 
Note that ${R^1_\rms}[n]$ is the secret rate in $n$-th time slot, which obtains from
\begin{align}
R^1_\rms[n]=\mathrm{max}\{R_\rmt-R^1_\rmw[n],0\}.
\end{align}
Note that the constraints \eqref{pso_1} and \eqref{dep_constraint} represent Bob's SCP and Willie's DEP at state $\rmH_1$, respectively. {The constraints \eqref{power_b_constraint_1} and \eqref{power_c_constraint} are related to Bob's and Carlo's beamformers.} Finally, constraints \eqref{velocity_constraint}, \eqref{acc_constraint}, \eqref{L_start_end}, \eqref{v_start_end}, and \eqref{fly_speed_step} concern the UAV's velocity, acceleration, and trajectory. 
Since the aforementioned multi-objective problem $\rmP0$ is complicated to solve, we use a coefficient $\kappa$ to weight Bob's average secret rate and $1-\kappa$ to weight Carlo's average covert rate, where $ \kappa\in (0,1)$. Inspired by \cite{marler2010weighted}, the multi-objective optimization function at each time slot $n\in[N]$ can be  transformed as
\begin{align}
&R\left(\bw_\mathrm{b,1}[n],\bw_\mathrm{c}[n],\bv[n],\ba[n],{\bL}_\rma[n]\right)\nn\\*
&=\kappa\frac{1}{N}\sum\limits_{n = 1}^N \!{{\rmP^{\mathrm{{ba,1}}}_\mathrm{sc}}[n]{R^1_\rms}[n]}
+\left(1-\kappa\right)\frac{1}{N}\sum\limits_{n = 1}^N \! {{\rmP^{\mathrm{{ca}}}_\mathrm{cc}}[n]{R_\rmc}} .
\end{align}
Therefore, the corresponding optimization problem under $\rmH_1$ for each $n\in[N]$ can be formulated as:
\begin{subequations}
\begin{align}  
{\rmP1}:\!\!\max_{\substack{\bw_\mathrm{b,1}[n],\!\bw_\mathrm{c}[n],\\\bv[n],\ba[n],{\bL}_\rma[n]}} \quad &\!\!  R\!\left(\bw_\mathrm{b,1}\![n],\bw_\mathrm{c}[n],\bv[n],\ba[n],{\bL}_\rma[n]\right)\label{obj_fun_P1}\\*
\mbox{s.t.}~~~~~~~~
&\eqref{fly_step},~\eqref{speed_step},~\eqref{pso_1}-\eqref{v_start_end}.
\end{align}
\end{subequations}
Thus, the multi-objective optimization problem $\rmP0$ is transformed into the single-objective optimization problem $\rmP1$, and a feasible solution will be given in Section III.}

Under state $\mathrm{H_0}$, only Bob transmits. The corresponding optimization problem for $\forall n\in[N]$ is formulated as follows:
\begin{subequations}
\begin{align}  
{ \rmP2}: \max_{\substack{\bw_\mathrm{b,0}[n],\\\bv[n],\ba[n],{\bL}_\rma[n]}} \quad& \frac{1}{N}\sum\limits_{n = 1}^N{{{\rmP^{\mathrm{{ba,0}}}_\mathrm{sc}}} [n]{R^0_\rms}[n]} \label{obj_fun_P2}\\*
\mbox{s.t.} \quad~~~~~
&{{\rmP^{\mathrm{{bw,0}}}_\mathrm{so}}}[n]\le \eta_s,  \label{pso_0}\\*
& \|\bw_\mathrm{b,0}[n]\|^2\le Q_\mathrm{b}^\mathrm{max},\label{power_b_constraint_0}    \\*     
&\eqref{fly_step},~\eqref{speed_step},~\eqref{velocity_constraint}-\eqref{v_start_end}.
\end{align}
\end{subequations}
where ${R^0_\rms}[n]$ is the target secret rate in $n$-th time slot, which is obtained in a similar way to ${R^1_\rms}[n]$. Note that objective function \eqref{obj_fun_P2} is Bob's average effective secret rate. The constraints \eqref{pso_0} and \eqref{power_b_constraint_0} are related to Bob's SOP and beamformer.
\vspace{4mm}
\section{Main Results} \label{sec:main_results}
\subsection{Derivations of Optimization Constraints}
To derive the closed forms of SCP, SOP, and CCP, we use the following Lemma \ref{lem-1} from~\cite{park2012outage}.
\begin{lemma} \label{lem-1}
For a Gaussian random vector $\bx\sim \mathcal{CN}({\bf\mu},{\bf\Sigma})$ with the eigen-decomposition of its covariance matrix $\bf\Sigma={\bf\Psi}\Lambda {\bf\Psi}^H$. For a given $\bf\bar Q $, the cumulative distribution function (CDF) of ${{\bx}^H}{\bf\bar Q}{\bx}$ is given by 
\begin{align}
{\rm Pr}\{{{\bx}^H}{\bf\bar Q}{\bf x}\leq\tau\}
\!=\!\frac{1}{ 2\pi}\!\int_{-\infty}^{\infty}\! {\frac{e^{\tau(j\omega+\beta)}}{j\omega+\beta}}\!\times \!\frac{e^{-c}}{\det({\bf I}\!+\!(j\omega\!+\!\beta){\bf Q})}d\omega , \label{eq}
\end{align}
for some $\beta\ge0$ such that ${\bf I}+\beta {\bf Q}$ is positive definite, where 
${\bf Q}={{\bf\Lambda} ^{H/2}}{{\bf\Psi} ^H}{\bf\bar Q} {\bf\Psi} {{\bf\Lambda} ^{{1}/2}}$, ${\bf\chi} ={{\bf\Lambda} ^{{ 1}/2}}{{\bf\Psi} ^H}\bmu $, and  $c = {{\bf\chi} ^H}{\left( {\bI + \frac{1}{{jw + \beta }}{{\bf\bar Q}^{ - 1}}} \right)^{ - 1}}{\bf\chi} $.
\end{lemma}

\subsubsection{Secrecy Connection Probability}
Bob' SCP at state $\rmH_1$ is given as
\begin{align}
\!{{\rmP^{\mathrm{{ba,1}}}_\mathrm{sc}}}[n]\!=&1-\mathrm{Pr} \left\{\bx_\mathrm{ca}^H[n]\bx_\mathrm{ca}[n]\geq\tau_\mathrm{ca}^1[n]\right\}\nn\\[3mm]
=&1\!\!-\!\text{exp}\!\!\left(\!\! -\frac{ \left| \bh^H_\mathrm{ba}[n]\bw_\mathrm{b,1}[n] \right|^2\!\!-\!\left(2^{R_\rmt}\!-\!1\right)\!\sigma _\mathrm{a}^2[n]  }{\left(2^{R_\rmt}-1\right)\left| \bh^H_\mathrm{ca}[n]\bw_\mathrm{c}[n] \right|^2}\!  \right)\!\!.\label{scp1_fun}
\end{align}
\begin{proof}
To obtain the result of $\mathrm{Pr} \left\{{\mathrm{SNR_{ba,1}}}[n]\! \le \!2^{R_\rmt}\!-\!1 \right\}$ in Eq. \eqref{psc1}, we define $\bx_\mathrm{ca}[n]=\bh^H_\mathrm{ca}[n]\bw_\mathrm{c}[n]$. It follows from Eqs. \eqref{SINR_b_1} and \eqref{psc1}, that $\bx_\mathrm{ca}[n]\sim \mathcal{CN}(0,\lambda^1_{\rm{ca}}\bI)$ and we further have 
\begin{align}
\bx_\mathrm{ca}^H[n]\bx_\mathrm{ca}[n]\geq   \frac{ \left| \bh^H_\mathrm{ba}[n]\bw_\mathrm{b,1}[n] \right|^2 }{2^{R_\rmt}-1} -\sigma_\mathrm{a}^2 [n]=:\tau_\mathrm{ca}^1[n], \label{taoco1}
\end{align}
based on {Lemma \ref{lem-1}} and let ${\bf\bar Q}=\bI$, it follows from Eqs. \eqref{eq} and \eqref{taoco1} that
\begin{align}
&{\rm Pr}\{{{\bf x}_\mathrm{ca}^H}[n]{\bf x}_\mathrm{ca}[n]\!\geq\!\tau_\mathrm{ca}^1[n]\}\nn\\[3mm]
&=1\!-\!\frac{1}{ 2\pi j}\int_{-\infty}^{\infty} \!\frac{e^{\tau_\mathrm{ca}^1[n] (j\omega+\beta)}}{ j\omega\!+\!\beta}\!\times\!\frac{1}{\det({\bf I}\!+\!(j\omega\!+\!\beta){\bf\Lambda})}\rmd\omega \label{27}\\[3mm]
&=\!1\!-\!\frac{1}{ 2\pi j}\int_{\beta-j\infty}^{\beta+j\infty} \frac{e^{\tau_\mathrm{ca}^1[n] s}}{s}\times \frac{1}{ 1+s\lambda ^1_{\rm{ca}}[n]}\rmd s \label{28}\\[3mm]
&=\!1\!-\!\frac{1}{{2\pi j}}\oint_{\mathcal C} {\frac{{{e^{\tau_\mathrm{ca}^1[n] s}}}}{s} \times \frac{1}{{{1 + s\lambda ^1_{\rm{ca}}[n]}}}\rmd s}, \label{contour_integration}
\end{align}
where Eq. \eqref{27} follows from Lemma \ref{lem-1}, Eq. \eqref{28} follows by setting $s=\beta+j\omega$, and in Eq. \eqref{contour_integration}, ${\cal C}$ denotes a contour of integration that includes the imaginary axis and the entire left half of the complex plane. 
And $\lambda ^1_{\rm{ca}}[n]=\left| \bh^H_\mathrm{ca}[n]\bw_\mathrm{c}[n] \right|^2$ is the eigenvalue of  $\bx_{\rm{ca}}[n]$. To derive the expression of Eq. \eqref{contour_integration}, we define 
\begin{align}
F(s)&=\frac{{{e^{\tau_\mathrm{ca}^1[n] s}}}}{s} \times \frac{1}{1 + s\lambda ^1_{\rm{ca}}[n] }. \label{residue_function}
\end{align}
It follows from Eq. \eqref{residue_function} that the residue of $F(s)$ at $s =0$ is  
\begin{align}
{\rm{Res}}\!\left(F(s),0\right) \!\!=\!\! \mathop {\lim }\limits_{s \to 0} s\!\! \times\!\!\frac{{{e^{\tau_\mathrm{ca}^1[n] s}}}}{{s(1\!\! +\!\! s\lambda ^1_{\rm{ca}}[n])}} 
\!\!= \!\!\mathop {\lim }\limits_{s \to 0} \frac{{{e^{\tau_\mathrm{ca}^1[n] s}}}}{{1 \!\!+\!\! s\lambda ^1_{\rm{ca}}[n]}} 
\!\!=\!\! 1.
\end{align}
{Furthermore, the residue of $F(s)$ at $s =-\frac{1}{\lambda ^1_{\rm{ca}}[n]}$ satisfies}
\begin{align}
{\rm{Res}}\!\left(\!\!F(s), - \frac{1}{\lambda ^1_{\rm{ca}}[n] }\!\right) &=\mathop {\lim }\limits_{s \to  - \frac{1}{\lambda ^1_{\rm{ca}}[n] }} \left(s + \frac{1}{\lambda ^1_{\rm{ca}}[n] }\right) \!\!\times \!\!\frac{{{e^{\tau_\mathrm{ca}^1[n] s}}}}{{s(1 + s\lambda ^1_{\rm{ca}} )}}\\
&=\mathop {\lim }\limits_{s \to  - \frac{1}{\lambda ^1_{\rm{ca}} }} \frac{{{e^{\tau_\mathrm{ca}^1[n] s}}}}{{s\lambda ^1_{\rm{ca}}[n] }}\\
&=- e^{ - \frac{\tau_\mathrm{ca}^1[n]} {\lambda ^1_{\rm{ca}}[n]} }.
\end{align}
Thus, 
\begin{align}
{\rm Pr}\{{{\bf x}_\mathrm{ca}^H[n]}{\bf x}_\mathrm{ca}[n]\geq\tau_\mathrm{ca}^1[n]\}&\!\!=\!\!1\!\!-\!\!\left( {\mathop {{\mathop{\rm Re}\nolimits}~ s}\limits_{s = 0} \!F\left( s \right) + \mathop {{\mathop{\rm Re}\nolimits}}\limits_{s = 1/\tau_\mathrm{ca}^1[n] } sF\left( s \right)} \right)\nn\\
&= e^{ - \frac{\tau_\mathrm{ca}^1[n]} {\lambda ^1_{\rm{ca}}[n]} }\!.\label{psc1_result}
\end{align}
Combine the above definitions of ${\lambda ^1_{\rm{ca}}}$, $\tau_\mathrm{ca}^1$, and Eq. (\ref{psc1_result}), we can obtain the result in Eq. \eqref{scp1_fun}.
\end{proof}
Analogously, Bob's SCP at state $\rmH_0$ can be given as
\begin{align}
{{\rmP^{\mathrm{{ba,0}}}_\mathrm{sc}}[n]}&=1-\mathrm{Pr} \left\{\bx_\mathrm{ca}^H[n]\bx_\mathrm{ca}[n]\geq\tau_\mathrm{ca}^0[n]\right\}\nn\\&=\text{exp}\left(\!\! -\frac{(2^{R_t}\!-\!1 ){\sigma _\mathrm{a}^2[n]}}{ \left| \bh^H_\mathrm{ba}[n]\bw_\mathrm{b,0}[n] \right|^2}  \right) .\label{scp0}
\end{align}

\subsubsection{Secret Outage Probability} 
Analogous to the derivation of SCP, Bob's SOP at state $\rmH_0$ and $\rmH_1$ can be respectively given as
\begin{align}
\rmP_\mathrm{so}^\mathrm{bw,0}[n]&=\text{exp}\left( -\frac{(2^{R_\rmw}-1 ){\sigma _\mathrm{w}^2[n]}}{\left| \bh^H_\mathrm{bw}[n]\bw_\mathrm{0}[n] \right|^2}  \right),\label{sop0}
\end{align}
\begin{align}
\rmP_\mathrm{so}^\mathrm{bw,1}[n]&=1-\text{exp}\!\!\left(\! -\frac{ \left| \bh^H_\mathrm{bw}[n]\bw_\mathrm{b,1}[n] \right|^2\!-\!\left(2^{R_\rmw}\!-\!1\right)\sigma _\mathrm{w}^2[n] }{\left(2^{R_\rmw}-1\right)\left| \bh^H_\mathrm{cw}[n]\bw_\mathrm{c}[n] \right|^2} \right)\label{sop1}.
\end{align}
\subsubsection{Covert Connection Probability} 
Based on the SIC scheme, Carlo's CCP equals Bob's SCP multiplied by Carlo's connection probability after Alice decodes Bob's secret message. Therefore, Carlo's CCP is given by
\begin{align}
{\rmP^{\mathrm{{ca}}}_\mathrm{cc}}[n]=&\left[1\!\!-\!\text{exp}\!\!\left(\! -\frac{ \left| \bh^H_\mathrm{ba}[n]\bw_\mathrm{b,1}[n] \right|^2-\left(2^{R_\rmt}\!\!-\!1\right)\!\sigma _\mathrm{a}^2[n]  }{\left(2^{R_\rmt}-1\right)\left| \bh^H_\mathrm{ca}[n]\bw_\mathrm{c}[n] \right|^2}\!  \right)\!\right]\nn\\
&\times\! \text{exp}\left( -\frac{(2^{R_\rmc}-1 ){\sigma _\mathrm{a}^2}[n]}{\left| \bh^H_\mathrm{ca}[n]\bw_\mathrm{c}[n] \right|^2}  \right).\label{ccp}
\end{align}
Note that the above formulation holds because the ${\mathrm{SNR_{ba,1}}[n]}$ and ${\mathrm{SNR_{ca}}[n]}$ are independent.
\subsubsection{Detection Performance at Willie} 
On Willie's side, there are two possible states exist for detecting Carlo's transmission, i.e., $\mathrm{H_0}$ and $\mathrm{H_1}$, where $\mathrm{H_0}$ represents the null state, indicating that Carlo is keeping silence, and $\mathrm{H_1}$ denotes Carlo is transmitting signals to Alice. Therefore, the variances of the received signals corresponding to these two states at Willie are
\begin{align}
\left\{\!\!\!\! \begin{array}{ll}
\Sigma_\mathrm{w}^0[n]\!\!:=\left|{\bh^H_\mathrm{bw}}[n]\bw_\mathrm{b,0}[n]\right|^2 + {\sigma_\mathrm{w}^2[n]},&\mathrm{H_0},\\
\\
\Sigma_\mathrm{w}^1[n]\!\!:=\!\!\left| {\bh^H_\mathrm{bw}}[n]\bw_\mathrm{b,1}[n]\right|^2 + \left| {\bh^H_\mathrm{cw}}[n]\bw_\mathrm{c}[n]\right|^2+ {\sigma_\mathrm{w}^2[n]},&\mathrm{H_1}.
\end{array} \right.
\end{align}
To detect whether Carlo is transmitting to Alice or not, Willie uses binary hypothesis testing. Thus, the optimal test that minimizes $\mathrm{P_e}[n]$ is the likelihood ratio test with $\psi^* = 1$ as the threshold~\cite{yan2018delay}, which is given by
\begin{align}
\frac{{{\mathbb P_\mathrm{1}} \buildrel \Delta \over = \prod\nolimits_{{\rm{i}} = 1}^m {f\left( {{z^i}[n]|{{\rm{H}}_1}} \right)} }}{{{\mathbb P_\mathrm{0}} \buildrel \Delta \over = \prod\nolimits_{{\rm{i}} = 1}^m {f\left( {{z^i}[n]|{{\rm{H}}_0}} \right)} }}\underset {\mathrm {H}_{0}}{\overset {\mathrm {H}_{1}}{\gtrless }}1,
\end{align}
where $f\left( {{z^i}[n]|{{\rm{H}}_1}} \right) = \mathcal{CN} (0,\Sigma_\mathrm{w}^1[n])$ and $f\left( {{z^i}[n]|{{\rm{H}}_0}} \right) = \mathcal{CN}(0,\Sigma_\mathrm{w}^0[n])$ are the likelihood functions of ${z^i}[n]$ at states $\mathrm{H_1}$ and $\mathrm{H_0}$, respectively.
Therefore, Willie's optimal decision rule to minimize the total error rate in $n$-th time slot is 
\begin{align}
{ T_\rmw}[n] \buildrel \Delta \over = \frac{1}{m}\sum\limits_i^m  {\left| {{z^i}[n]} \right|^2}\underset {\mathrm {H}_{0}}{\overset {\mathrm {H}_{1}}{\gtrless }} {Q_\mathrm{th}}[n],
\end{align}
where the test statistic $T_\rmw[n]$ is the average power of each received symbol at Willie and ${Q^*_\mathrm{th}}[n]$ is the optimal threshold for $T_\rmw[n]$ in $n$-th time slot, which is given by
\begin{equation} 
{Q^*_\mathrm{th}}[n] = \frac {\Sigma_\mathrm{w}^1[n] \Sigma_\mathrm{w}^0[n]}{\Sigma_\mathrm{w}^1[n]-\Sigma_\mathrm{w}^0[n]} \ln \left (\frac {\Sigma_\mathrm{w}^1[n]}{\Sigma_\mathrm{w}^0[n]}\right)\!. 
\end{equation}
Note that the radiometer with ${Q^*_\mathrm{th}}[n]$ is indeed the optimal detector when Willie knows the likelihood functions exactly (i.e., when there are no nuisance parameters embedded in the likelihood functions). Note that $T_\rmw[n]$ is a chi-squared random variable with $2n$ degrees of freedom, the likelihood functions of $T_\rmw[n]$ at states $\mathrm{H_0}$ and $\mathrm{H_1}$ are respectively given by 
\begin{align} 
f({ T_\rmw}[n]| \mathrm{H}_{0})=&\frac {T_\rmw^{m-1}[n]}{\Gamma (m)}\left ({\frac {m}{\Sigma_\mathrm{w}^0[n]}}\right)^{m} e^{-\frac {m{T_\rmw}[n]}{\Sigma_\mathrm{w}^0[n]}}, 
\end{align}
\begin{align} 
f({ T_\rmw}[n]| \mathrm {H}_{1})=&\frac {T_\rmw^{m-1}[n]}{\Gamma (m)}\left ({\frac {m}{\Sigma_\mathrm{w}^1[n]}}\right)^{m} e^{-\frac {m{T_\rmw}[n]}{\Sigma_\mathrm{w}^1[n]}}, 
\end{align}
where $\Gamma (m)=(m-1)!$ is the Gamma function. Hence, the false alarm probability $\mathrm{P_F}[n]$ and the miss detection probability $\mathrm{P_M}[n] $ in Eq. \eqref{dep} are respectively given by
\begin{align} 
\mathrm{P_F}[n]=&\Pr ({ T_\rmw}[n] > {Q_\mathrm{th}}[n] | \mathrm {H}_0) \nn\\[3mm]
=&\int_{{Q_\mathrm{th}}[n]}^\infty \frac {T_\rmw^{m-1}[n]}{\Gamma (m)}\left ({\frac {m}{\Sigma_\mathrm{w}^0[n]}}\right)^{m} e^{-\frac {m{T_\rmw}[n]}{\Sigma_\mathrm{w}^0[n]}}\rmd{ T_\rmw}[n] \\
=& 1- \frac {\gamma \left ({m, \frac {m {Q_\mathrm{th}}[n] }{\Sigma_\mathrm{w}^0[n]}}\right)}{\Gamma (m)},\label{44}
\end{align}
and
\begin{align} 
\mathrm{P_M}[n]=&\Pr ({ T_\rmw}[n] < {Q_\mathrm{th}}[n] | \mathrm {H}_1) \nn\\*
=& \int_{0}^{Q_\mathrm{th}[n]} \frac {T_\rmw^{m-1}[n]}{\Gamma (m)}\left ({\frac {m}{\Sigma_\mathrm{w}^1[n]}}\right)^{m} e^{-\frac {m{T_\rmw}[n]}{\Sigma_\mathrm{w}^1[n]}}  \rmd { T_\rmw}[n]\\*
=& \frac {\gamma \left ({m, \frac {m {Q_\mathrm{th}}[n] }{\Sigma_\mathrm{w}^1[n]}}\right)}{\Gamma (m)},\label{45} 
\end{align}
where 
\begin{align} 
\gamma \left ({m,x}\right)\!\!=\!\!\int _{0}^{x} e^{-t} t^{m-1} dt\!=\!\left(m -1\right)!\left[ {1 \!\!-\!\! {e^{ - x}}\sum\limits_{k = 0}^{m - 1} {\frac{{{x^k}}}{{k!}}} } \right]\!\!.
\end{align}
For simpsity, we define $\psi_0\!:=\!\frac {m {Q_\mathrm{th}}[n] }{\Sigma_\mathrm{w}^0[n]}$ and $\psi_1\!:=\!\frac {m {Q_\mathrm{th}}[n] }{\Sigma_\mathrm{w}^1[n]}$.
The total error detection probability of Willie can be rewritten as 
\begin{align}
\mathrm{P_e}[n]  &= 1- \frac {\gamma \left ({m, \frac {m {Q_\mathrm{th}}[n] }{\Sigma_\mathrm{w}^0[n]}}\right)}{\Gamma (m)}  + \frac {\gamma \left ({m, \frac {m {Q_\mathrm{th}}[n] }{\Sigma_\mathrm{w}^1[n]}}\right)}{\Gamma (m)}\\
&=1-{e^{ -\psi_1}}\sum\limits_{k = 0}^{m - 1} {\frac{{{\psi_1^k}}}{{k!}}}+{e^{ -\psi_0}}\sum\limits_{k = 0}^{m - 1} {\frac{{{\psi_0^k}}}{{k!}}}.\label{DEP}
\end{align}

Based on the secret and covert communication performance metrics, i.e., SCP and SOP for the secret user Bob, CCP and DEP for the covert user Carlo, and UAV's flight constraints, we will give the solution to the optimization problems.

\subsection{Solutions to Optimization Problems}\label{sec:solutions}
Based on the aforementioned analysis and results, the objective function of $\rmP1$ can be expressed as Eq. \eqref{Opt_obj}, which is on the top of the next page. Note that the problem $\rmP1$ is challenging to solve due to the complicated objective function and the constraints. To address this, a SCA-BCD iterative algorithm is proposed to solve the optimization problems approximately using results from the previous subsection. Specifically, to solve the optimization problem $\rmP1$, we first optimize the secret and covert beamformers for any given UAV's trajectory. Subsequently, we optimize the UAV's velocity, acceleration, and flight trajectory based on the optimized secret and covert beamformers. We solve the optimization problem $\rmP1$, and problem $\rmP2$ can be solved similarly.
\begin{figure*}[ht]
\begin{small}
\begin{align}
\!\!R\!\left(\bw_\mathrm{b,1}[n],\bw_\mathrm{c}[n],\bv[n],\ba[n],{\bL}_\rma[n]\right)
\!=&\frac{\kappa}{N}\sum\limits_{n = 1}^N\!{\left[\!1\!\!-\!\text{exp}\!\!\left(\!\! -\frac{ \left| \bh^H_\mathrm{ba}[n]\bw_\mathrm{b,1}[n] \right|^2\!\!-\!\left(2^{R_\rmt}\!\!-\!\!1\right)\!\sigma _\mathrm{a}^2[n]  }{\left(2^{R_\rmt}-1\right)\left| \bh^H_\mathrm{ca}[n]\bw_\mathrm{c}[n] \right|^2}\!  \right)\!\right]\!\!{R^1_\rms}[n]}\nn\\*
&+ \frac{(1\!-\!\kappa)}{N}\!\sum\limits_{n = 1}^N\!\!\Bigg\{\!\!\!\left[\!1\!\!-\!\text{exp}\!\!\left(\! -\frac{ \left| \bh^H_\mathrm{ba}[n]\bw_\mathrm{b,1}[n] \right|^2\!-\!\left(2^{R_\rmt}\!-\!1\right)\!\sigma _\mathrm{a}^2[n]  }{\left(2^{R_\rmt}\!-\!1\right)\left| \bh^H_\mathrm{ca}[n]\bw_\mathrm{c}[n] \right|^2}\!  \right)\!\right]\!\!\times\! \text{exp}\!\left(\!\! -\frac{(2^{R_\rmc}\!-\!1 ){\sigma _\mathrm{a}^2}[n]}{\left| \bh^H_\mathrm{ca}[n]\bw_\mathrm{c}[n] \right|^2}  \!\right)\!\!\!\Bigg\}{R_\rmc}. \label{Opt_obj}
\end{align} 
\end{small}
\hrulefill
\end{figure*}
\subsubsection{\textbf{Secret and Covert Beamforming Design}}
Given the UAV's velocity $\bv[n]$, acceleration $\ba[n]$, and trajectory $\bL_\rma[n]$, the objective function $R\!\left(\bw_\mathrm{b,1}[n],\bw_\mathrm{c}[n],\bv[n],\ba[n],{\bL}_\rma[n]\right)$ in Eq. \eqref{Opt_obj} degenerates to $R\left(\bw_\mathrm{b,1}[n],\bw_\mathrm{c}[n]\right)$. Thus, the optimization problem $\rmP1$ degenerates as
\begin{subequations}
\begin{align}  
{\rmP1.1}: \max_{\bw_\mathrm{b,1}[n],\bw_\mathrm{c}[n]} \quad&   R\left(\bw_\mathrm{b,1}[n],\bw_\mathrm{c}[n]\right)\label{func_P1.1}\\*
\mbox{s.t.}~~~~~\quad
&\eqref{pso_1},~\eqref{dep_constraint},~\eqref{power_b_constraint_1},~\eqref{power_c_constraint}.
\end{align}
\end{subequations}
Note that problem ${\rmP1.1}$ is challenging to solve due to the complicated objective function \eqref{func_P1.1} and the covertness constraint \eqref{dep_constraint}. Consistent with the previous study~\cite{ma2021robust}, we consider the case of perfect covertness, which corresponds to $\rmP_\rme=1$, i.e., $\epsilon=0$. Note that perfect covertness can be achieved in multi-user collaborative transmission since we consider that Willie's location is available to Bob and Carlo, and considering the perfect covertness can effectively analyze the boundary performance of the proposed scheme. Furthermore, when the UAV trajectory and the flight parameters are fixed, the optimization problem equals maximizing the objective function in each time slot. To address the proposed optimization problems, we adopt the SDR technique~\cite{luo2010semidefinite} in a fixed time slot. Thus, we have
\begin{align}
{h}_\mathrm{\dagger a} &= {\bf{h}}_\mathrm{\dagger a} ^H{\bf{h}}_\mathrm{\dagger a},~{h}_\mathrm{\dagger w}  = {\bf{h}}_\mathrm{\dagger w} ^H{\bf{h}}_\mathrm{\dagger w},\label{sdr_H}\\
{{\bf{W}}_\rmc} &= {{\bf{w}}_\rmc}{\bf{w}}_\rmc^H \Leftrightarrow {{\bf{W}}_\rmc} \succ 0,~{\rm{rank}}\left( {{{\bf{W}}_{\rm{c}}}} \right){\rm{ = 1}},\label{sdr_wc}\\
{{\bf{W}}_\mathrm{b,1}} &= {{\bf{w}}_\mathrm{b,1}}{\bf{w}}_\mathrm{b,1}^H \Leftrightarrow {{\bf{W}}_\mathrm{b,1}} \succ 0,~{\rm{rank}}\left( {{{\bf{W}}_\mathrm{b,1}}} \right){\rm{ = 1}},
\\
{{\bf{W}}_\mathrm{b,0}} &= {{\bf{w}}_\mathrm{b,0}}{\bf{w}}_\mathrm{b,0}^H \Leftrightarrow {{\bf{W}}_\mathrm{b,0}} \succ 0,~{\rm{rank}}\left( {{{\bf{W}}_\mathrm{b,0}}} \right){\rm{ = 1}}.\label{sdr_wb}
\end{align}
In the case of perfect covertness, the constraint \eqref{dep_constraint} satisfies
\begin{align}  
Q_\mathrm{b}^\mathrm{max}\mathrm{Tr}\left( {\bf{h}}_{{\rm{bw}}}^H{{\bf{h}}_{{\rm{bw}}}}\right)\!&\nn\\
=&\mathrm{Tr}\!\left(\bh^H_\mathrm{bw}\!\bW_\mathrm{b,1}\bh_\mathrm{bw} \!\right)\!\!+\!\!\mathrm{Tr}\!\left(\bh^H_\mathrm{cw}\!\bW_\mathrm{c}\bh_\mathrm{cw}\!\right)\!.\label{perfect_covert}
\end{align}
Note that the SOP constraint \eqref{pso_1} holds when the inequality takes the equal, thus we have
\begin{align}  
\!R^1_\rmw\!=\!{\rm log}\!\left(\!1\!+\!\frac{Q_\mathrm{b}^\mathrm{max}\mathrm{Tr}\!\left(  {\bf{h}}_{{\rm{bw}}}^H{{\bf{h}}_{{\rm{bw}}}}\right)\!-\!\mathrm{Tr}\left(\bh^H_\mathrm{cw}\!\bW_\mathrm{c}\bh_\mathrm{cw} \right)}{\sigma_{\rm w}-\ln\! {\eta _\rms}\mathrm{Tr}\left(\bh^H_\mathrm{cw}\bW_\mathrm{c}\bh_\mathrm{cw}\right)}\right)\!.
\end{align}
By substituting Eqs. \eqref{sdr_H}--\eqref{perfect_covert} into $R\left(\bw_\mathrm{b,1}[n],\bw_\mathrm{c}[n]\right) $ in \eqref{func_P1.1}, the objective function in the fixed UAV's trajectory for a determined time slot can be rewritten as $R\left({\bf{W}}\!_{\rm{c}} \right)$ in Eq. \eqref{obj_fun_P1.2} on the top of next page.
\begin{figure*}[!ht]
\begin{small}
\begin{align} 
R\left({\bf{W}}\!_{\rm{c}} \!\right)
=& \kappa\! \left[ {1 \!-\! {\rm{exp}}\!\left( { \!\!-\!\! \frac{Q_\mathrm{b}^\mathrm{max}\mathrm{Tr}\!\left({\bf{h}}_{{\rm{ba}}}^H{{\bf{h}}_{{\rm{ba}}}}\right)\!-\!\mathrm{Tr}\!\left(\bh^H_\mathrm{cw}\!\bW_\mathrm{c}\bh_\mathrm{cw}\right)\!-\!{\left( {{2^{{R_\rmt}}} - 1} \right)\sigma _{\rm{a}}^2}}{{\left( {{2^{{R_\rmt}}} - 1} \right){\mathrm{Tr}}\left(\bh^H_\mathrm{ca}\!\bW_\mathrm{c}\bh_\mathrm{ca} \right)}} } \right)} \!\right]\!\times\! \left[ {{R_\rmt} \!-\! {\rm{log}}\!\left( {1 + \frac{{Q_{\rm{b}}^\mathrm{max}{\mathrm{Tr}}\left(\bh^H_\mathrm{bw}\bh_\mathrm{bw}\right)\! -\! {\mathrm{Tr}}\left( \bh^H_\mathrm{cw}\!\bW_\mathrm{c}\bh_\mathrm{cw}\right)}}{{{\sigma _{\rm{w}}} - {\mathrm{Tr}}\left(\bh^H_\mathrm{cw}\!\bW_\mathrm{c}\bh_\mathrm{cw} \right)}}} \right)} \!\right]\nn\\*
& + \left( {1\! -\! \kappa } \right)\!\left\{\! \left[\! {1 \!- \!{\rm{exp}}\!\left(\! { - \frac{Q_\mathrm{b}^\mathrm{max}\mathrm{Tr}\!\left(\bh^H_\mathrm{bw}\bh_\mathrm{bw}\!\right)\!-\!\mathrm{Tr}\!\left(\bh^H_\mathrm{cw}\bW_\mathrm{c}\bh_\mathrm{cw}\!\right)\!-\!{\left(\! {{2^{{R_\rmt}}}\! -\! 1} \right)\!\!\sigma _{\rm{a}}^2}}{{\left( {{2^{{R_\rmt}}} - 1} \right){\mathrm{Tr}}\left(\bh^H_\mathrm{ca}\bW_\mathrm{c}\bh_\mathrm{ca} \right)}} } \!\right)\!}\right] 
\times  {\rm{exp}}\left(  - \frac{{({2^{{R_\rmc}}} - 1)\sigma _{\rm{a}}^2}}{{{\mathrm{Tr}}\left(\bh^H_\mathrm{ca}\!\bW_\mathrm{c}\bh_\mathrm{ca} \right)}} \right) \right\}{R_\rmc}.\label{obj_fun_P1.2}
\end{align}
\end{small}
\hrulefill
\end{figure*}
\begin{lemma}\label{lemma2}
The objective function $R\left({\bf{W}}_{\rm{c}} \right)$ exhibits at least one inflection point $\mathrm{Tr}\left( {\bf{W}}_{\rm{c}}\right)^*$ in the interval $(0, \infty)$.
\end{lemma} 
\begin{proof}
Since $R\left({\bf{W}}_{\rm{c}} \right)$ is too complicated to directly prove its monotonicity. The inflection point can be analyzed as follows.\\
Case 1: when $\mathrm{Tr}\left( {\bf{W}}_{\rm{c}}\right) = 0$, the sum rate $R(\bW_\mathrm{c})=\kappa \rmP_{{\rm{sc}}}^1R^1_\rms$ due to zero covert rate.
\\
Case 2: when $\mathrm{Tr}\left( {\bf{W}}_{\rm{c}}\right) \to \infty $, the sum rate $R(\bW_\mathrm{c})=0$ due to the $100\%$ connection outage.\\
Case 3: when $\mathrm{Tr}\left( {\bf{W}}_{\rm{c}}\right)$ increases from $0$, the sum rate $R(\bW_\mathrm{c})$ also increases. The proof is as follows, when $\mathrm{Tr}\left( {\bf{W}}_{\rm{c}}\right)\to {0^+}$, it is denoted by $\mathrm{Tr}\left( {\bf{W}}_{\rm{c}}\right)^{0+}$. Thus, the objective function only includes the secret rate, and it can be expressed as
\begin{align}
\mathop {\lim }\limits_{{\mathrm{Tr}}\left( {{{\bf{W}}_{\rm{c}}}} \right) \to {{\rm{0}}^{\rm{ + }}}} R\left( {{\mathrm{Tr}}\left( {{{\bf{W}}_{\rm{c}}}} \right)} \right) = \kappa \rmP_{{\rm{sc}}}^1\left( {{\mathrm{Tr}}\left( {{{\bf{W}}_{\rm{c}}}} \right)^{0+}} \right){R_\rms},
\end{align}	
when $\mathrm{Tr}\left( {\bf{W}}_{\rm{c}}\right)^{0+}$ increases $\Delta\mathrm{Tr}\left( {\bf{W}}_{\rm{c}}\right)$, the objective function can be expressed as
\begin{align}
\mathop {\lim }\limits_{\Delta{\mathrm{Tr}}\left( {{{\bf{W}}_{\rm{c}}}} \right) \to {{\rm{0}}^{\rm{ + }}}} R&\left( {{\mathrm{Tr}}\left( {{{\bf{W}}_{\rm{c}}}} \right)^{0+}\!+\!\Delta {\mathrm{Tr}}\left( {{{\bf{W}}_{\rm{c}}}} \right)} \right)\nn \\*
=&\kappa \rmP_{{\rm{sc}}}^1\left( {{\mathrm{Tr}}\left( {{{\bf{W}}_{\rm{c}}}} \right)^{0+}+\Delta {\mathrm{Tr}}\left( {{{\bf{W}}_{\rm{c}}}} \right)} \right){R_{\rm{s}}}\nn\\*
&\!\!+ \!\left( 1\!-\!\kappa  \right)\!\rmP_{{\rm{sc}}}^1\!\!\left( \!{{\mathrm{Tr}}\left( \!{{{\bf{W}}_{\rm{c}}}} \right)^{0+}\!\!+\!\!\Delta {\mathrm{Tr}}\left( {{{\bf{W}}_{\rm{c}}}} \right)} \!\right)\!{R_{\rm{c}}}.
\end{align}	
Thus the increase of the sum of the secure rate is as follows
\begin{small}
\begin{align}
\mathop {\lim }\limits_{\Delta {\mathrm{Tr}}\left( {{{\bf{W}}_{\rm{c}}}} \right) \to {\rm{0}}} &R\left( {{\mathrm{Tr}}\left( {\bf{W}}_{\rm{c}} \right)^{0+}\!+\!\Delta {\mathrm{Tr}}\left( {{{\bf{W}}_{\rm{c}}}} \right)} \right)\!-\!\mathop {\lim }\limits_{{\mathrm{Tr}}\left( {{{\bf{W}}_{\rm{c}}}} \right) \to {{\rm{0}}^{\rm{ + }}}} R\left( {{\mathrm{Tr}}\left( {{{\bf{W}}_{\rm{c}}}} \right)} \right)\nn\\
&=\left( {{\rm{1 - }}\kappa } \right)\rmP_{{\rm{sc}}}^1\left( {{\mathrm{Tr}}\left( {{{\bf{W}}_{\rm{c}}}} \right)^{0+}\!+\!\Delta {\mathrm{Tr}}\left( {{{\bf{W}}_{\rm{c}}}} \right)} \right){R_{\rm{c}}} \ge 0.
\end{align}	
\end{small}Since the objective function initially increases and eventually decreases to zero, there is at least one inflection point in the objective function. Thus, this lemma is proved.
\end{proof}

Based on the SDR technology and ignoring the rank-one constraints, the optimization problem $\rmP1.1$ transforms to $\rmP1.2$,
\begin{subequations} \label{P1.2}
\begin{align} 
{\rmP1.2}: \max_{\bW_\mathrm{c}} \quad&   R\left( {\bf{W}}_{\rm{c}} \right) \label{objfun1.2}\\*
\mbox{s.t.} \quad
&{\mathrm{Tr}}\left(\bW_\mathrm{c}\right)\le  Q_\mathrm{c}^\mathrm{max},\label{2.1.2.b}\\*
&\bW_\mathrm{c}\succ 0.\label{2.1.2.c}
\end{align}
\end{subequations}
\begin{thm}\label{thm1}
The optimal solution to $\mathrm{P}1.2$ is given as
\begin{align}
\bW_\mathrm{c}^*= \mathop {\arg \max }\limits_{{{\bf{W}}_{{\rm{c}}}\in{\cal W}_{\rmc}}}R\left( {{{\bf{W}}_{{\rm{c}}}}} \right),
\end{align}	
where ${\cal W}_{\rmc}=\{\bW_{\rmc}|\eqref{2.1.2.b},\eqref{2.1.2.c}\}$, the secret beamformer $\bW_\mathrm{c}^*$ and covert beamformer $\bW_\mathrm{b,1}^*$ can be obtained from Eq. \eqref{perfect_covert} by singular value decomposition (SVD) or Gaussian randomization procedure (GRP).
\end{thm} 
\begin{proof}
From Lemma \ref{lemma2}, the optimal rate satisfies $\frac{{\rmd R\left( {{{\bf{W}}_{{\rm{c}}}}} \right)}}{{\rmd{\mathrm{Tr}}\left( {{{\bf{W}}_{{\rm{c}}}}} \right)}} = 0$. {Thus, we adopt a binary search algorithm (BSA) to find the optimal transmit power ${\mathrm{Tr}}\left(\bW_\mathrm{c}\right)^*$ for Carlo and use the perfect covertness constraint Eq. \eqref{perfect_covert} to obtain the optimal transmit power ${\mathrm{Tr}}\left(\bW_\mathrm{b,1}\right)^*$ for Bob.} The BSA is shown in Algorithm \ref{alg_1}. Furthermore, ${{{\bf{W}}^*_{\rm{c}}}}$ and ${\bf{W}}_{\rm{b,1}}^*$ can be obtained with SVD. Therefore, the optimal beamformers $\mathbf{w}_{\mathrm{b,1}}$ and $\mathbf{w}_{\mathrm{c}}$ can be obtained using SVD when $\text{rank}(\mathbf{W}_{\mathrm{b,1}}^*)\!\!=\!\!1$ and $\text{rank}(\mathbf{W}_{\mathrm{c}}^*)\!\!=\!\!1$, i.e., ${{\bf{W}}^*_\mathrm{b,1}}\!\! = \!\!{{\bf{w}}_\mathrm{b,1}}{\bf{w}}_\mathrm{b,1}^H$ and ${{\bf{W}}^*_\rmc}= {{\bf{w}}_\rmc}{\bf{w}}_\rmc^H$.
When $\text{rank}(\mathbf{W}_{\mathrm{b,1}}^*)\!\!>\!\!1$ and $\text{rank}(\mathbf{W}_{\mathrm{c}}^*)\!\!>\!\!1$, the GRP is used to produce a high-quality solution~\cite{luo2010semidefinite}.
\end{proof}
\begin{algorithm}
\setstretch{1}
\caption{Proposed BSA to solve the optimization problem $\rmP1.2$\label{alg_1}.}
\renewcommand{\algorithmicensure}{\textbf{Output:}}
\begin{algorithmic}[1]
\STATE\textbf{Input}: $\bL_\rma$, $\bL_\rmb$, $\bL_\rmw$, $\bL_\rmc$, $Q_\rmc^\mathrm{max}$, $R_\rmt$, $R_\rmc$, $\eta_\rms$, $\rmP_\rme=1$, and the pre-defined accuracy $\zeta_1 > 0$;
\STATE \textbf{Output}: $ \mathrm{Tr}\left( {{{\bf{W}}_{\rm{c}}}} \right)^*$;
\STATE Initialize $q_{0}=0$, $q_{1}=Q_\mathrm{c}^\mathrm{max}$, and calculate the first order derivation of ${R}\left( {{{\bf{W}}_{{\rm{c}}}}} \right)$ respective to ${\mathrm{Tr}}\left({{{\bf{W}}_{{\rm{c}}}}} \right)$;
\WHILE{$q_{1}-q_{0}\ge \zeta_1$}
\STATE set $q_\rmm=\frac{q_{0}+q_{1}}{2}$;
\IF{ ${R}\left( q_\rmm \right)>0$}
\STATE $q_{0}=q_\rmm$;
\ELSE
\STATE $q_{1}=q_\rmm$;
\ENDIF
\ENDWHILE
\end{algorithmic}
\end{algorithm}
Note that $\rmP1.2$ can obtain a high-quality solution for the secret and covert beamformers optimization. Furthermore, the optimal secret and covert beamformers in each time slot can be obtained based on the aforementioned method.

\subsubsection{\textbf{Trajectory, velocity, and acceleration Optimization}}  
{The optimized secret beamformer ${{\bf{w}}_\mathrm{b,1}}[n]$ and covert beamformer ${{\bf{w}}_\mathrm{c}}[n]$ of $\rmP1.1$ can be obtained by solving $\rmP1.2$.} Thus given the above beamformers, the objective function $R\!\left(\bw_\mathrm{b,1}[n],\bw_\mathrm{c}[n],\bv[n],\ba[n],{\bL}_\rma[n]\right)$ in Eq. \eqref{Opt_obj} degenerates to $R\left(\bv[n],\ba[n],{\bL}_\rma[n]\right)$. Therefore, the optimization problem for UAV's velocity, acceleration, and trajectory can be given as
\begin{subequations}
\begin{align}  
{\rmP1.3}: \max_{\bv[n],\ba[n],{\bL}_\rma[n]} \quad&  R\left(\bv[n],\ba[n],{\bL}_\rma[n]\right)\label{60a}\\*
\mbox{s.t.} ~~~~~\quad &\eqref{fly_step},~\eqref{speed_step},~\eqref{velocity_constraint}-\eqref{v_start_end}.
\end{align}
\end{subequations}
{The optimization problem $\rmP1.3$ is complicated to solve due to the non-convex objective function and the velocity constraint. To address the non-convex objective function in Eq. \eqref{60a}, we obtained an approximation to the original problem by transforming the problem of maximizing the sum of secret and covert rates in each time slot.
}

Since we consider that the number of time slots is large enough, the period of each time slot is sufficiently small such that the location of the UAV is considered to be approximately unchanged within each time slot. Thus, the UAV's velocity and acceleration can be considered constant within each time slot but vary between time slots. Since the air-to-ground path-loss exponents are set to $\xi_\mathrm{\dagger a}= -2$, the channel gains from the ground users to the UAV and warden respectively given as 
\begin{align} 
\mathrm{Tr}\!\left(\bh^H_\mathrm{\dagger w}[n]\bh_\mathrm{\dagger w}[n] \right)&={N_{\rm{\dagger}}}{h_{{\rm{\dagger w}}}}[n],\\
\mathrm{Tr}\!\left(\bh^H_\mathrm{\dagger a}[n]\bh_\mathrm{\dagger a}[n] \right)&={N_{\rm{\dagger}}}{h_{{\rm{\dagger a}}}}\left({{\bf{L}}_{\rma}}[n]\right)\![n]\\*
&={N_{\rm{\dagger}}}{\lambda _0}\left(d_{{\rm{\dagger a}}}[n]\right)^{\xi_\mathrm{\dagger a}}\\*
&= \frac{{N_{\rm{\dagger}}}{{\lambda _0}}}{{{H^2} + {{\left\| {{{\bf{L}}_\rma}[n] - {{\bf{L}}_\dagger}[n]} \right\|}^2}}}.
\end{align} 

Note that the objective function in $\rmP1.3$ is still non-convex, and it is hard to solve. In order to achieve a convex approximation of the objective function, we further proceed as follows
\begin{align}
{\mu _\mathrm{ba}}[n] = {\left\| {{{\bf{L}}_\rma}[n] - {{\bf{L}}_\rmb}[n]} \right\|^2},\label{u_ba}\\
{\mu _\mathrm{ca}}[n] = {\left\| {{{\bf{L}}_\rma}[n] - {{\bf{L}}_\rmc}[n]} \right\|^2}.\label{u_ca}
\end{align}
Therefore, the objective function in $\rmP1.3$ can be rewritten as $R\!\left({\mu _\mathrm{ba}}[n],\!{\mu _\mathrm{ca}}[n] \right)$ in Eq. \eqref{new_obj_traj}.
\begin{figure*}[!t]
\begin{small}
\begin{align} 
\!R\!\left({\mu _\mathrm{ba}}[n],\!{\mu _\mathrm{ca}}[n] \right)\! =& \kappa\!\! \left[ {1\!\! - \!{\rm{exp}}\!\!\left(\!\! { - \frac{{{N_{\rm{b}}}{h_{{\rm{ba}}}}({\mu _\mathrm{ba}}[n]){\mathrm{Tr}}\left( {{{\bf{W}}_{\rm{c}}}[n]} \right)\!\! -\!\! \left( {{2^{{R_\rmt}}}\!\! -\!\! 1} \right)\sigma _{\rm{a}}^2[n]}}{{{\mathrm{Tr}}\left( {{{\bf{W}}_{\rm{c}}[n]}} \right)\left( {{2^{{R_\rmt}}} - 1} \right){N_{\rm{c}}}h _\mathrm{ca}({\mu _\mathrm{ca}}[n])}}} \right)} \!\!\right]\!\!\times\!\! \left[\! {{R_\rmt}\!\! -\!\! {\rm{log}}\left( {\!\!1\!\! +\!\! \frac{{{N_{\rm{b}}}{h_\mathrm{bw}}Q_{\rm{b}}^\mathrm{max} \!\!- \!\!{N_{\rm{c}}}{h_\mathrm{cw}}{\mathrm{Tr}}\left( {{{\bf{W}}_{\rm{c}}[n]}} \right)}}{{{\sigma _{\rm{w}}}[n] - {N_{\rm{c}}}{h_\mathrm{cw}}{\rm{ln}}{\eta _\rms}{\mathrm{Tr}}\left( {{{\bf{W}}_{\rm{c}}[n]}} \right)}}}\!\! \right)}\! \right]\nn\\
&\!\!+\!\! \left( {1 \!\!- \!\!\kappa } \right)\!\!\left\{\!\! \left[\! {1\!\! - {\rm{exp}}\!\!\left( \!\!{ - \frac{{{N_{\rm{b}}}{h_{{\rm{ba}}}}({\mu _\mathrm{ba}}[n]){\mathrm{Tr}}\left( {{{\bf{W}}_{\rm{c}}[n]}} \right)\!\! -\!\! \left( {{2^{{R_\rmt}}} - 1} \right)\sigma _{\rm{a}}^2[n]}}{{{\mathrm{Tr}}\left( {{{\bf{W}}_{\rm{c}}[n]}} \right)\left( {{2^{{R_\rmt}}} - 1} \right){N_{\rm{c}}}h_\mathrm{ca}({\mu _\mathrm{ca}}[n])}}} \right)}\!\! \right]\!\!\times\!{\rm{exp}}\!\left( \!\!{ - \frac{{({2^{{R_\rmc}}}\!\! -\!\! 1)\sigma _{\rm{a}}^2[n]}}{{{\mathrm{Tr}}\left( {{{\bf{W}}_{\rm{c}}[n]}} \right){N_{\rm{c}}}h_\mathrm{ca}({\mu _\mathrm{ca}}[n])}}} \right)\!\! \right\}\!\!{R_\rmc}\label{new_obj_traj}
\end{align}
\end{small}
\hrulefill
\end{figure*}
Note that the objective function is a binary function with respect to ${\mu _{{\rm{ba}}}}[n]$ and ${\mu _{{\rm{ca}}}}[n]$. Given the reference UAV's trajectory ${\bf{L}}_{\rm{a}}^{\rm{r}}[n]$ for $\forall n\in[N]$ of the , i.e., $\mu _{{\rm{ba}}}^{\rm{r}}[n]$ and $\mu _{{\rm{ca}}}^{\rm{r}}[n]$, we can solve the optimization problem $\rmP1.3$ approximately using the SCA method. The first-order Taylor expansion of $R\!\left({\mu _\mathrm{ba}}[n],\!{\mu _\mathrm{ca}}[n] \right)$ in Eq. \eqref{new_obj_traj} is defined as
\begin{align}
\!\!\!\widetilde R\!\left( {{\mu _{{\rm{ba}}}}[n],\!{\mu _{{\rm{ca}}}}[n]} \right)\!\! :=& R\left( {\mu _{{\rm{ba}}}^{\rm{r}}[n],\mu _{{\rm{ca}}}^{\rm{r}}[n]} \right) \nn\\
&\!\!+\!\! \frac{{\partial \!R\!\left( {{\mu _{{\rm{ba}}}}[n],\!{\mu _{{\rm{ca}}}}[n]} \right)}}{{\partial {\mu _{{\rm{ba}}}}[n]}}\!\left( {{\mu _{{\rm{ba}}}}[n] \!\!-\!\! \mu _{{\rm{ba}}}^{\rm{r}}[n]} \right)\nn\\
&\!\!+\!\! \frac{{\partial\! R\!\left( {{\mu _{{\rm{ba}}}}[n],\!{\mu _{{\rm{ca}}}}[n]} \right)}}{{\partial {\mu _{{\rm{ca}}}}[n]}}\!\left( {{\mu _{{\rm{ca}}}}[n] \!\!- \!\!\mu _{{\rm{ca}}}^{\rm{r}}[n]} \right)\!,
\end{align}
where $\frac{{\partial{ R\left( {{\mu _{{\rm{ba}}}}[n],{\mu _{{\rm{ca}}}}[n]} \right)}}}{{\partial {\mu _{{\rm{ba}}}}[n]}}$ and $\frac{{\partial{ R\left( {{\mu _{{\rm{ba}}}}[n],{\mu _{{\rm{ca}}}}[n]} \right)}}}{{\partial {\mu _{{\rm{ca}}}}[n]}}$ are the partial derivative of $R\left( {{\mu _{{\rm{ba}}}}[n],{\mu _{{\rm{ca}}}}[n]} \right)$ respect to $\mu _{{\rm{ba}}}[n]$ and $\mu _{{\rm{ca}}}[n]$, respectively. Since the minimum velocity constraint \eqref{velocity_constraint} is still non-convex, thus given the reference velocity vector ${{\bf{v}}^{\rm{r}}}[n]$, it can be approximated using the first-order Taylor expansion
\begin{align} 
v_\mathrm{min} \le{{\bf{v}}^{\rm{r}}}[n] + 2{(\bv^{\rm r}}[n])^T\left( {{\bf{v}}[n]{\rm{ - }}{{\bf{v}}^{\rm{r}}}[n]} \right).
\end{align}
Therefore, the problem $\rmP1.3$ can be approximated as
\begin{subequations}\begin{align}  
{\rmP1.4}: \!\!\max_{\substack{\bv[n],\ba[n],{\bL}_\rma[n],\\{\mu _\mathrm{ba}}[n],{\mu _\mathrm{ca}}[n]}} \quad   \!\!\!\! &\widetilde R\left( {{\mu _{{\rm{ba}}}}[n],{\mu _{{\rm{ca}}}}[n]} \right)\\*
\mbox{s.t.} ~~~~\quad
&\eqref{fly_step},\eqref{speed_step},\eqref{velocity_constraint}-\eqref{v_start_end}, \eqref{u_ba}, \eqref{u_ca}.
\end{align}
\end{subequations}
Note that the above optimization problem is convex, thus it can be solved by the CVX tools~\cite{cvx}. By solving the optimization problem $\rmP1.4$, we can obtain a high-quality feasible solution of the UAV's trajectory, velocity, and acceleration for optimization problem $\rmP1.3$.

\subsubsection{\textbf{Overall Algorithm}}
We propose a BCD algorithm to solve the original optimization $\rmP1$ problem by integrating the SDR-based BSA and the SCA algorithm. Specifically, the SCA-BCD algorithm process is given in Algorithm \ref{algo2}.
\begin{algorithm}[t] 
\caption{Proposed SCA-BCD algorithm to approximately solve the original optimization problem $\rmP1$.}\label{algo2}
\renewcommand{\algorithmicensure}{\textbf{Output:}}
\begin{algorithmic}[1]
\STATE \textbf{Input}: $\bL_\rma[0]$, $\bL_\rma[N]$, $\bL_\rmb$, $\bL_\rmw$, $\bL_\rmc$, $Q_\rmb^\mathrm{max}$, $Q_\rmc^\mathrm{max}$, $R_\rmt$, $R_\rmc$, $\eta_\rms$, $\rmP_\rme$, $\mu _{{\rm{ba}}}^{\rm{0}}[n]$, $\mu _{{\rm{ca}}}^{\rm{0}}[n]$, $\bv^{\rm{0}}[n]$, $\ba^{\rm{0}}[n]$, and the pre-defined accuracy $\zeta_2 > 0$.
\STATE \textbf{Output}: $\widetilde R\left( {{\mu^{*} _{{\rm{ba}}}}[n],{\mu^{*} _{{\rm{ca}}}}[n]} \right)$, ${{{\bf{w}}^*_{\rm{b,1}}}}[n]$, ${\bf{w}}^*_{\rm{c}}[n]$, ${\bv}^*[n]$, ${\ba}^*[n]$, and ${\bL}^*_{\rm{a}}[n]$.
\WHILE{$\left| {{{\widetilde R}}\left( {{\mu^k _{{\rm{ba}}}}[n],{\mu^k _{{\rm{ca}}}}[n]} \right) \!\!-\!\! {{\widetilde R}}\left( {{\mu^{k - 1} _{{\rm{ba}}}}[n],{\mu^{k - 1} _{{\rm{ca}}}}[n]} \right)} \right|\!\! \ge\!\! \zeta_2$}
\STATE Optimizing the secret and covert beamformers:\\
Given $\bv^{k - 1}[n]$, $\ba^{k - 1}[n]$, and $\bL^{k - 1}_\rma[n]$, solve $\rmP1.2$ to update ${{{\bf{w}}^k_{\rm{b,1}}}}[n]$ and ${\bf{w}}^k_{\rm{c}}[n]$, respectively.\\
\STATE Optimizing the UAV's  velocity, acceleration, and trajectory:\\
Given the secret beamformer ${{{\bf{w}}^k_{\rm{b,1}}}}[n]$ and covert beamformer ${\bf{w}}^k_{\rm{c}}[n]$, solve $\rmP1.4$ to update ${\bf{L}}_{\rm{a}}^{k}[n]$, $\bv^{k}[n]$, and $\ba^{k}[n]$, respectively.
\STATE $k=k+1$;\\
\ENDWHILE
\end{algorithmic}
\end{algorithm} 
Furthermore, we can obtain the solution to the benchmark problem $\mathrm{P2}$ using the proposed SCA-BCD algorithm.

\subsubsection{\textbf{Complexity Analysis}}  
Note that Algorithm \ref{alg_1} is based on SDR and BSA, the computational complexity of Algorithm \ref{alg_1} is  ${\cal O}\!\left(\left(N_\rmb+N_\rmc\right)^4\sqrt{N_\rmb+N_\rmc}\mathrm{log}(1/\zeta_1) \right)$~\cite{luo2010semidefinite,cvx}. In addition, the computational complexity of the SCA algorithm is ${\cal O}\left(IN^3\right)$, where $I$ is the iteration number. Therefore, the total computational complexity of the SCA-BCD Algorithm \ref{algo2} is ${\cal O}\left(I^2N^3+I\left(N_\rmb+N_\rmc\right)^4\sqrt{N_\rmb+N_\rmc}\mathrm{log}(1/\zeta_1) \right)$. Analogously, the computational complexity of the solution to $\rmP2$ is ${\cal O}\left(I^2N^{3}+I\left(N_\rmb\right)^4\sqrt{N_\rmb}\mathrm{log}(1/\zeta_1)\right)$.

\section{Numerical Results}\label{sec:simulation}
In this section, we present extensive numerical examples to illustrate our results in Lemma \ref{lemma2} and Theorem \ref{thm1}, and the proposed SCA-BCD algorithm, where the key numerical parameters are set as follows. The number of antennas at Bob and Carlo are $N_\rmb=2$ and $N_\rmc=2$, respectively. The original and terminal positions of Alice are $\bL_\rma[0]=(0,0,200)$ and $\bL_\rma[N]=(500,500,200)$, and the positions of Bob, Carlo, and Willie are $\bL_\rmb=(200,300)$, $\bL_\rmw=(100,400)$, $\bL_\rmc=(200,150)$, respectively. The maximum transmit powers of Bob and Carlo is $Q_\rmb^\mathrm{max}=30$ dBm and $Q_\rmc^\mathrm{max}= 0$ dBm, respectively. The target secret rates of Bob and Carlo are $R_\rmt=16$ bps and $R_\rmc=4$ bps, and the SOP and DEP are $\eta_\rms=0.01$ and $\rmP_\rme=1$. The path-loss exponent of the ground channel and the air-to-ground are $\xi_\mathrm{ba}=\xi_\mathrm{ca}=-2$ and $\xi_\mathrm{bw}=\xi_\mathrm{cw}=-3$, respectively. Both the path-loss reference at 1 meter and the excessive path-loss coefficient are --10 dB~\cite{chen2021uav_tcom,bai2024efficient}. The Rician factors between Bob and Alice, and between Carlo and Alice are $K_\mathrm{ba} =3$ dB and $K_\mathrm{ca}= 0$ dB, respectively. The noise power at Alice and Willie are $\sigma _\rma^2=\sigma _\rmw^2=-90$ dBm~\cite{yan2021optimal,xu2022covert}. Specifically, we illustrate the numerical results of Bob's SCP and SOP at states $\rmH_0$ and $\rmH_1$, respectively. {We give the numerical results of Corlo's CCP and Willie's DEP at state $\rmH_1$.} Furthermore, we exhibit the effect of channel Rician factors on the sum of the secret and covert transmission rates. Finally, to demonstrate the performance of the proposed algorithm, we illustrate the optimal transmit powers of legitimate users and the UAV's flight parameters.

\begin{figure}[t]
\centering
\includegraphics[width=1\linewidth]{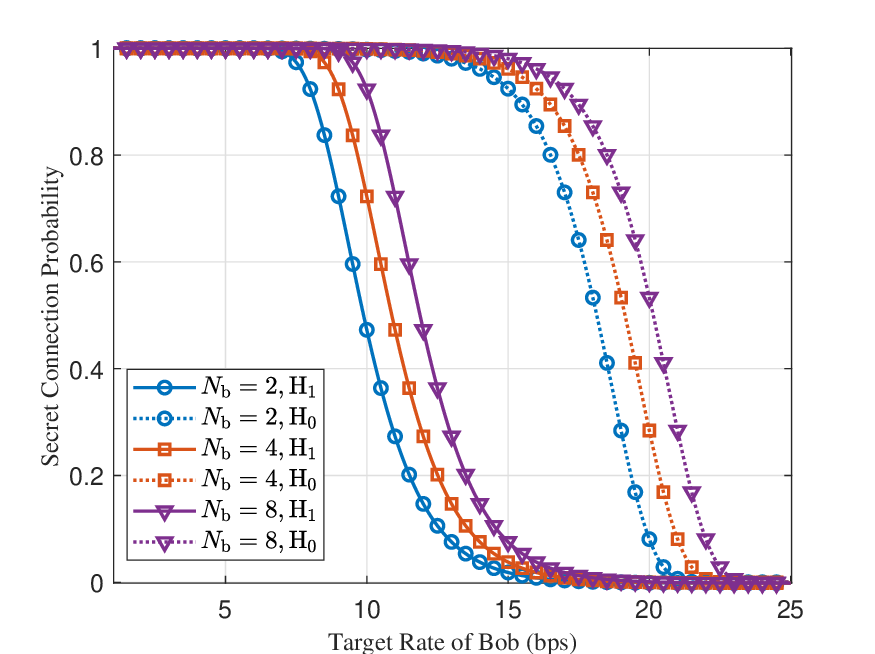}
\caption{Numerical plots of Bob's SCPs versus Bob's target rate $R_\rmt$ at UAV's original position at sates $\rmH_0$ and $\rmH_1$ with different number of antennas for Bob, where $N_\rmc=2$, $Q_\rmb^\mathrm{max}=30$ dBm, $Q_\rmc^\mathrm{max}=0$ dBm, $\xi_\mathrm{ba}=\xi_\mathrm{ca}=-2$, $K_\mathrm{ba}=3$ dB, $K_\mathrm{ca}= 0$ dB, and $\sigma _\rma^2=-90$ dBm.}
\label{fig:scp} 
\end{figure}
\begin{figure}[t]
\centering
\includegraphics[width=0.98\linewidth]{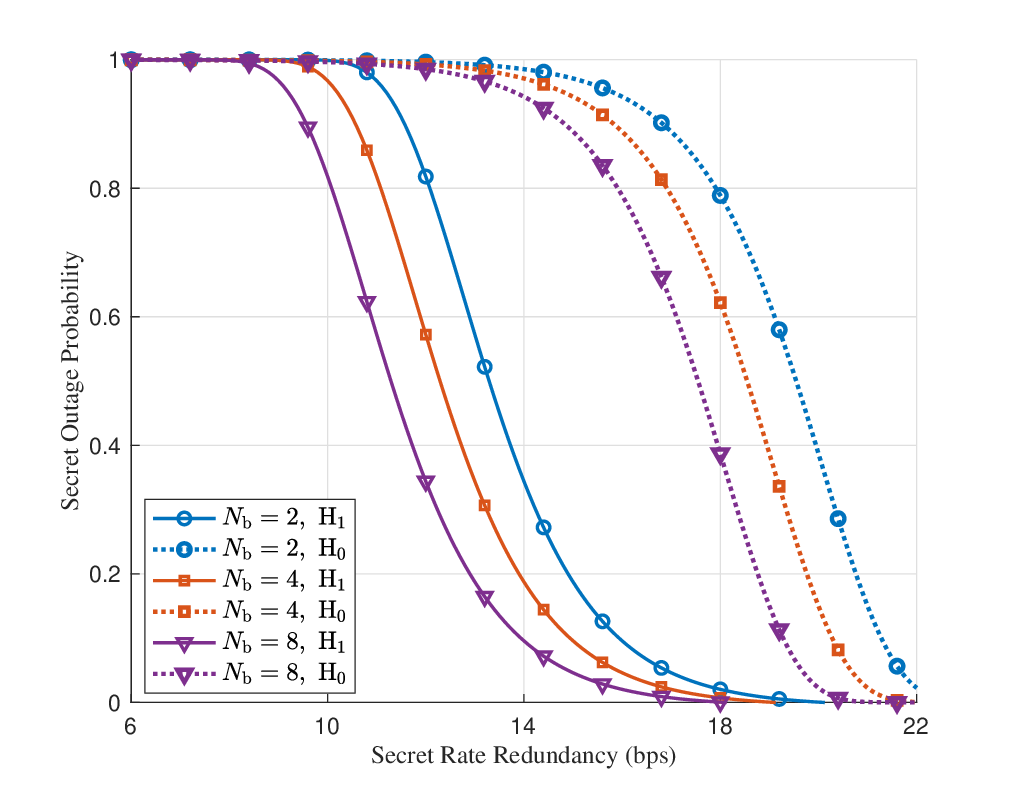}
\caption{Numerical plots of Bob's SOPs versus Bob's secret capacity redundancy at sates $\rmH_0$ and $\rmH_1$ with different number of antennas, where $N_\rmc=2$, $Q_\rmb^\mathrm{max}\!=\!30$ dBm, $Q_\rmc^\mathrm{max}\!=\!0$ dBm,  $\xi_\mathrm{bw}\!=\!\xi_\mathrm{cw}\!=\!-3$, and $\sigma _\rmw^2\!=\!-90$ dBm.}
\label{fig:sop} 
\end{figure}
\begin{figure}[t]
\centering
\includegraphics[width=1\linewidth]{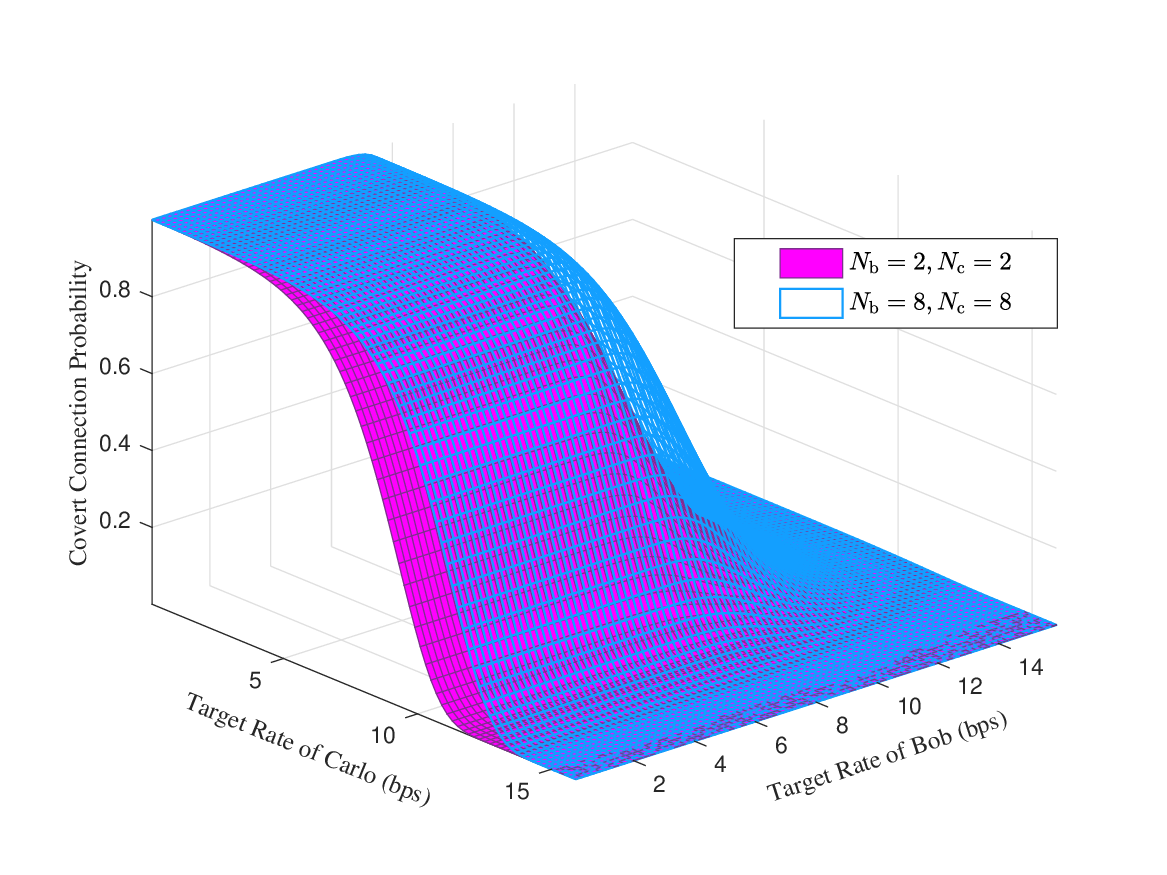}
\caption{Numerical plots of Carlo's CCPs versus Carlo's target rate $R_\rmc$ at UAV's original position with different number of antennas of Bob and Carlo, where $Q_\rmb^\mathrm{max}=30$ dBm, $Q_\rmc^\mathrm{max}= 0$ dBm, $\xi_\mathrm{ba}=\xi_\mathrm{ca}=-2$, $K_\mathrm{ba} =3$ dB, $K_\mathrm{ca}=0$ dB, and $\sigma _\rma^2=-90$ dBm.}
\label{fig:ccp} 
\end{figure}
\begin{figure}[t]
\centering
\includegraphics[width=0.9875\linewidth]{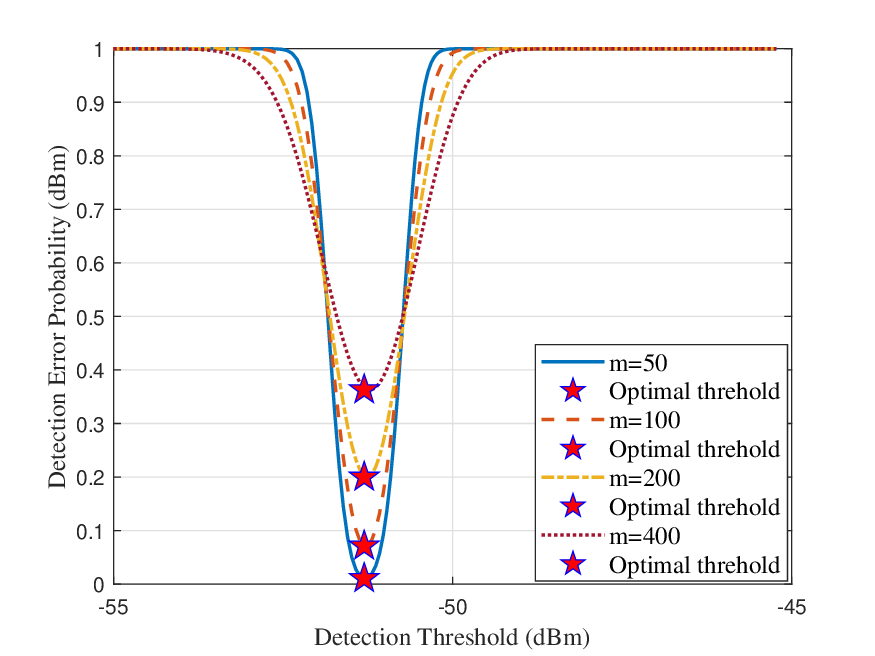}
\caption{Numerical plots of Willie's DEP versus Willie's detection threshold for different block lengths, where $Q_\rmb^\mathrm{max}=30$ dBm, $Q_\rmc^\mathrm{max}=0$ dBm, $\xi_\mathrm{bw}=\xi_\mathrm{cw}=-3$, and $\sigma _\rmw^2=-90$ dBm.
\label{fig:dep}} 
\end{figure}

To illustrate the impact of Bob's target transmission rate and the number of antennas on the SCP, we plot Bob's SCPs in Eqs. \eqref{scp1_fun} and \eqref{scp0} under both states $\rmH_1$ and $\rmH_0$ as a function of Bob's target rate $R_\rmt$ at UAV's original position in Fig. \ref{fig:scp} for different number of antennas at Bob. It is observed that as the target rate of Bob increases, the SCPs under both states $\rmH_0$ and $\rmH_1$ decrease from 1 to 0, which implies that communication outage occurs when the target rate becomes large under certain channel conditions. Particularly, it is observed that when the target transmission rate increases from 0, the SCP values under both $\rmH_1$ and $\rmH_0$ states show little variation. This is because, at this stage, the target transmission rate of the secret user Bob is low, and the SNR of the signal received by the UAV from Bob is sufficient to ensure reliable connection. However, once the target transmission rate reaches a certain threshold, the SCP values under both $\rmH_1$ and $\rmH_0$ states experience a steep decline. This is due to the fact that the UAV's received SNR from Bob's signal becomes insufficient to maintain a reliable connection, leading to the disruption of the confidential user's link. Furthermore, we can observe that the SCP at state $\rmH_0$ is larger than that of the state $\rmH_1$ at a specific target rate of Bob. This is because Carlo's signal is concealed within the Bob's signal at state $\rmH_1$, resulting in a lower SNR of Aclie compared to state $\rmH_0$. Note that increasing the number of antennas will improve the SCP at both states $\rmH_0$ and $\rmH_1$. 

In Fig. \ref{fig:sop}, in order to illustrate the impact of secret rate redundancy and the number of antennas on the SOPs formulated in Eqs. \eqref{sop0} and \eqref{sop1}, we plot Bob's SOPs as a function of the secret rate redundancy under both states $\rmH_0$ and $\rmH_1$. From Fig. \ref{fig:sop}, we can observe that as Bob's secret rate redundancy increases, the SOPs under both states $\rm H_0$ and $\rm H_1$ decrease. This phenomenon indicates that as the secret rate redundancy increases, it becomes more difficult for the adversary, Willie, to achieve a channel capacity that exceeds the secret rate redundancy. Consequently, the secret outage probability decreases, as the adversary's ability to intercept and decode the secret signal becomes increasingly limited. Similar to the SCPs, we can observe that the SOPs at state $\rmH_0$ are larger than that of the state $\rmH_1$ at a specific secret rate redundancy of Bob. This is because Carlo's signal is hidden within Bob's signal, leading to a lower SNR for Willie under $\rmH_1$ compared to $\rmH_0$. Furthermore, it can be observed that in both states $\rmH_0$  and $\rmH_1$, the secret outage probability decreases as the number of transmission antennas increases. This confirms that the use of multiple antennas for precoding at the transmitter effectively reduces the probability of data being intercepted.

In Fig. \ref{fig:ccp}, we plot Carlo's CCP  in Eq. \eqref{ccp} as a function of Bob's target rate $R_\rmt$ and Carlo's target rate $R_\rmc$ at UAV's original position. From Fig. \ref{fig:ccp}, we can observe that as Bob's and Carlo's target rates increase, the CCP decreases from 1 to 0, which indicates that covert connection outages will occur when Bob's or Carlo's target rates become large under specific channel conditions. This phenomenon includes two aspects of explanations, firstly, as Bob's target rate increases, the communication link between Bob and Alice is more likely to be outage, leading to a connection outage between Carlo and Alice. Second, even if reliable communication is maintained between Bob and Alice, a larger Carlo's target rate will result in a connection outage between Carlo and Alice. Similarly, we can observe that increasing the number of antennas will improve the CCP due to the diversity gain provided by multiple antennas.

In Fig. \ref{fig:dep}, we plot Willie's DEP versus the detection threshold for different block lengths. It can be seen from Fig. \ref{fig:dep} that the DEP initially decreases and then increases as the detection threshold increases, indicating the presence of an optimal threshold and the minimum DEP value. Note that the minimum DEP decreases as the signal sequence length increases, as Willie can make more informed judgments with larger sequences.
This observation is consistent with the theoretical analysis in Eq. \eqref{DEP}. Moreover, it suggests that perfect covertness can be achieved through collaboration between Bob and Carlo, corresponding to $\mathrm{P_e} = 1$.

\begin{figure}[t]
\centering
\includegraphics[width=1\linewidth]{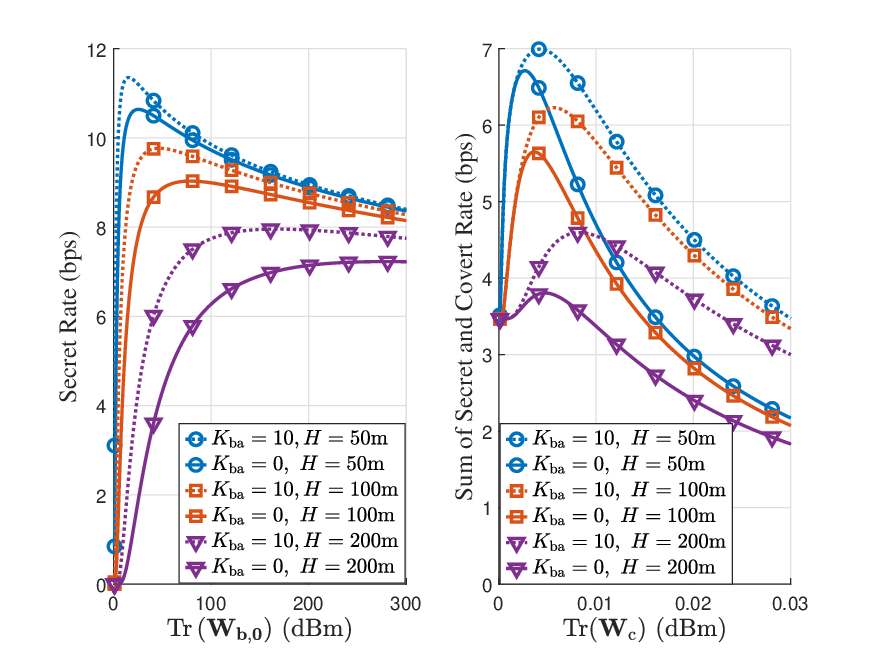}
\caption{Numerical plots of secret rate and the sum of the secret and covert rates versus Bob's and Carlo's transmit powers at the UAV's original position, respectively, with different Rician factors and UAV altitudes, where $Q_\rmb^\mathrm{max}=30$ dBm, $Q_\rmc^\mathrm{max}= 0$ dBm, $R_\rmt=16$ bps, $R_\rmc=4$ bps, $\xi_\mathrm{ba}=\xi_\mathrm{ca}=-2$, and $\sigma _\rma^2=-90$ dBm.}
\label{fig:R0_R1} 
\end{figure}

\begin{figure}[h]
\centering
\includegraphics[width=1\linewidth]{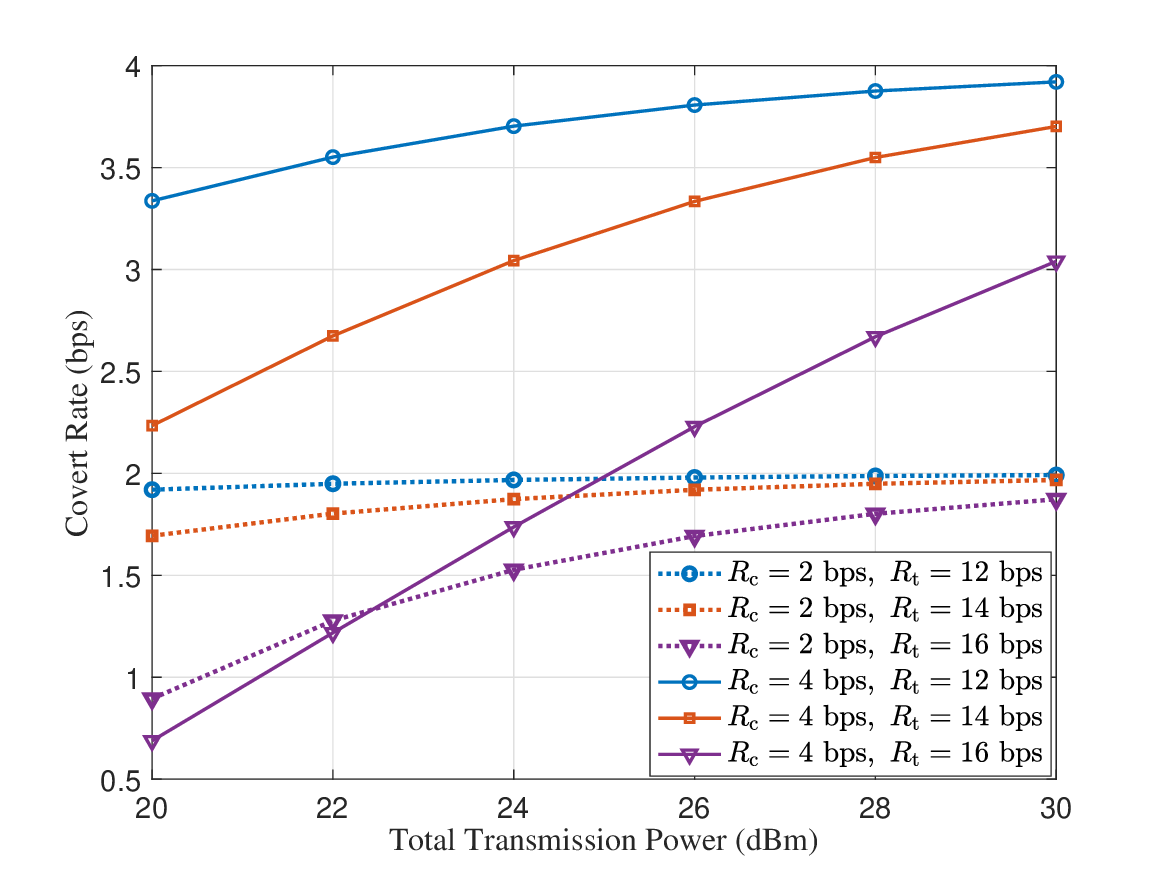}
\caption{Numerical plots of covert rate versus total transmission power at the UAV's original position with different target rates of Bob and Carlo, where $N_\rmc=2$, $\xi_\mathrm{ba}=\xi_\mathrm{ca}=-2$, $\xi_\mathrm{bw}=\xi_\mathrm{cw}=-3$, $K_\mathrm{ba}=3$ dB, $K_\mathrm{ca}= 0$ dB, and $\sigma _\rma^2=-90$ dBm.}
\label{fig:secret}
\end{figure}

\begin{figure}[t]
\centering
\includegraphics[width=1\linewidth]{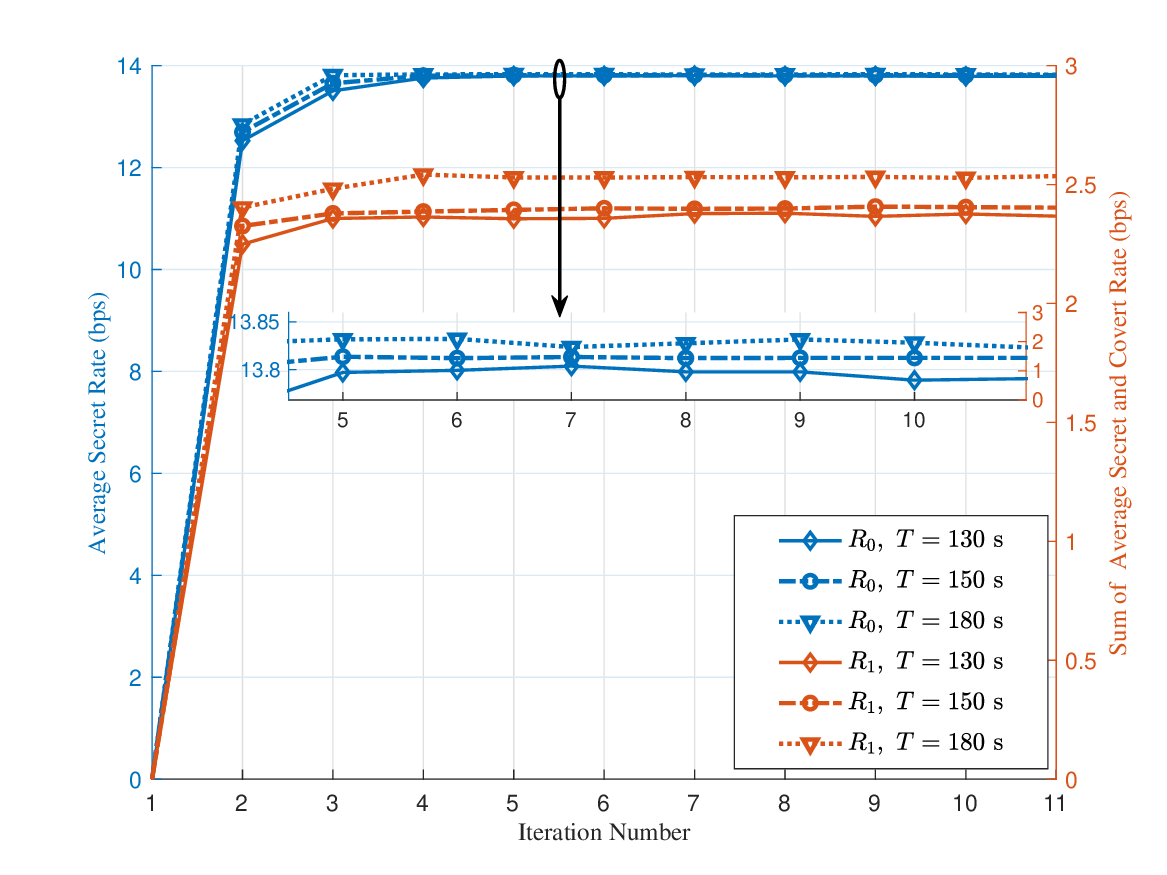}
\caption{Numerical plots  Bob's average secret and the sum of Bob's average secret and Carlo's average covert rates versus the iteration number at states $\rmH_0$ and $\rmH_1$, where $Q_\rmb^\mathrm{max}=30$ dBm, $Q_\rmc^\mathrm{max}\!=\!0$ dBm, $R_\rmt=16$ bps, $R_\rmc=4$ bps, $\kappa\!=\!0.5$, $\xi_\mathrm{ba}\!=\!\xi_\mathrm{ca}\!=\!-2$, $\xi_\mathrm{bw}\!=\!\xi_\mathrm{cw}\!=\!-3$, and $\sigma _\rma^2=\sigma _\rmw^2=-90$ dBm.}
\label{fig:avesumratevsit_H0_H1} 
\end{figure}

To verify the results in Lemma \ref{lemma2} and Theorem \ref{thm1}, we plot Bob's secret rate and the sum of secret and covert rates as a function of Bob and Carlo' transmit powers in Fig. \ref{fig:R0_R1}. Specifically, on the left part of Fig. \ref{fig:R0_R1}, we illustrate Bob's secret rate under state $\rmH_0$ for different channel Rician factors and UAV altitudes as a function of Bob's transmit power. As observed, Bob's secret rate increases initially and then decreases with the transmit power, indicating the existence of an optimal transmit power for Bob and an optimal beamformer. On the right part of Fig. \ref{fig:R0_R1}, we plot the sum of Bob's secret and Carlo's covert rates under state $\rmH_1$ for different Rician factors and UAV altitudes as a function of the transmit power of Carlo. Similarly, these curves indicate the existence of optimal transmit power and beamformer for both Bob and Carlo. The above numerical results verify our theoretical analysis in Lemma \ref{lemma2} and Theorem \ref{thm1}, and demonstrate the direction for power optimization and beamforming design.

In Fig. \ref{fig:secret}, we plot Carlo's covert rate versus the total transmission power for different secret and covert target rates. From Fig. \ref{fig:secret}, it can be seen Carlo's covert rate increases with the total transmission power because, as the total power increases, both Bob and Carlo can achieve higher probabilities of secure and covert connections, thereby increasing the covert rate. Furthermore, when the total transmission power is low, choosing a smaller target rate for Bob enables Carlo to achieve a higher covert rate. This is because a reliable connection for Bob is the foundation for Carlo's covert communication, and at lower total power levels, using a lower target rate ensures a more reliable connection. Furthermore, it can be observed that when the total transmission power is low, if Bob adopts a higher target rate for communication, Carlo can benefit from using a smaller target rate to achieve a higher covert transmission rate. This is because, at lower total power levels, the reduced connection probability of Bob directly causes a decrease in Carlo's connection probability. In this case, Carlo's adoption of a smaller target rate may lead to a greater benefit.

In Figs. \ref{fig:avesumratevsit_H0_H1} and \ref{fig:avesumratevsit}, we demonstrate the convergence performance of the proposed algorithm. Specifically, in Fig. \ref{fig:avesumratevsit_H0_H1}, we plot Bob's average secret and the sum of Bob's average secret and Carlo's average covert rates at states $\rmH_0$ and $\rmH_1$, respectively. Clearly, Bob's average secret at state $\rmH_0$ is significantly larger than that of the sum of Bob's average secret and Carlo's average covert rates at state $\rmH_1$. In Fig. \ref{fig:avesumratevsit}, for different flight periods, we plot the sum of the average secret and covert rates for two algorithms: JOTB, which Jointly Optimized the UAV's Trajectory and ground user's Beamformers, and SOTFB, which Solely Optimized UAV's Trajectory with Fixed Beamformers for ground users where both the secret and covert beamformers are identity matrices with fixed transmit powers. The numerical results demonstrate remarkable performance of our proposed algorithm and validate the performance of joint optimization of UAV's trajectory and users' beamformer designs over separate optimizations. Furthermore, it is observed that as the UAV's flight period increases, its sum of average secret and average covert rates becomes higher, and the convergence of the average secure rate is slightly slower. This is because, during longer flight periods, the search space for optimizing the UAV's flight trajectory expands, leading to a slower convergence. However, the extended flight time enables better flight trajectory decisions, thereby resulting in a higher corresponding average secure rate.

\begin{figure}[t]
\centering
\includegraphics[width=1\linewidth]{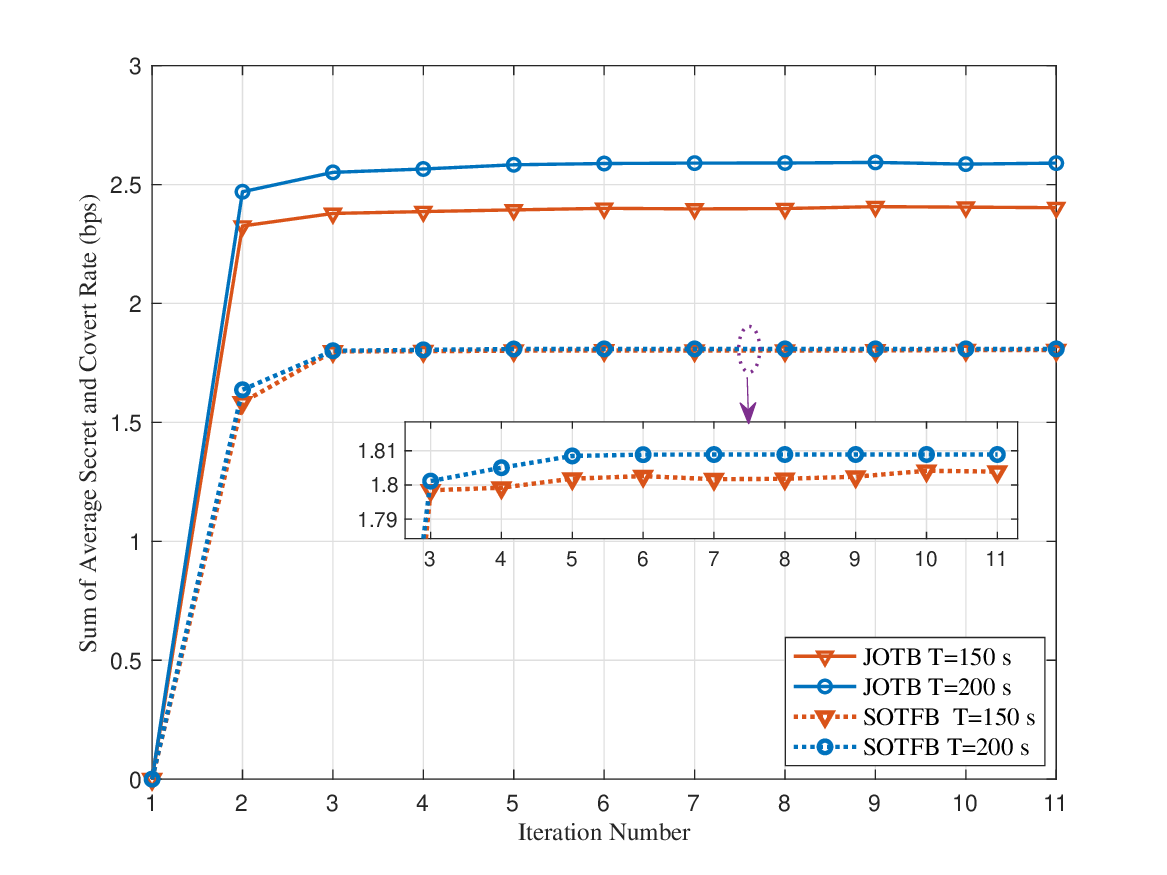}
\caption{Numerical plots of the sum of Bob's average secret and Carlo's average covert rates versus the iteration number, where $Q_\rmb^\mathrm{max}=30$ dBm, $Q_\rmc^\mathrm{max}=0$ dBm, $R_\rmt=16$ bps, $R_\rmc=4$ bps, $\kappa=0.5$, $\xi_\mathrm{ba}=\xi_\mathrm{ca}=-2$, $\xi_\mathrm{bw}=\xi_\mathrm{cw}=-3$, and $\sigma _\rma^2=\sigma _\rmw^2=-90$ dBm.}
\label{fig:avesumratevsit}
\end{figure}

\begin{figure}[t] 
\centering 
\includegraphics[width=1\linewidth]{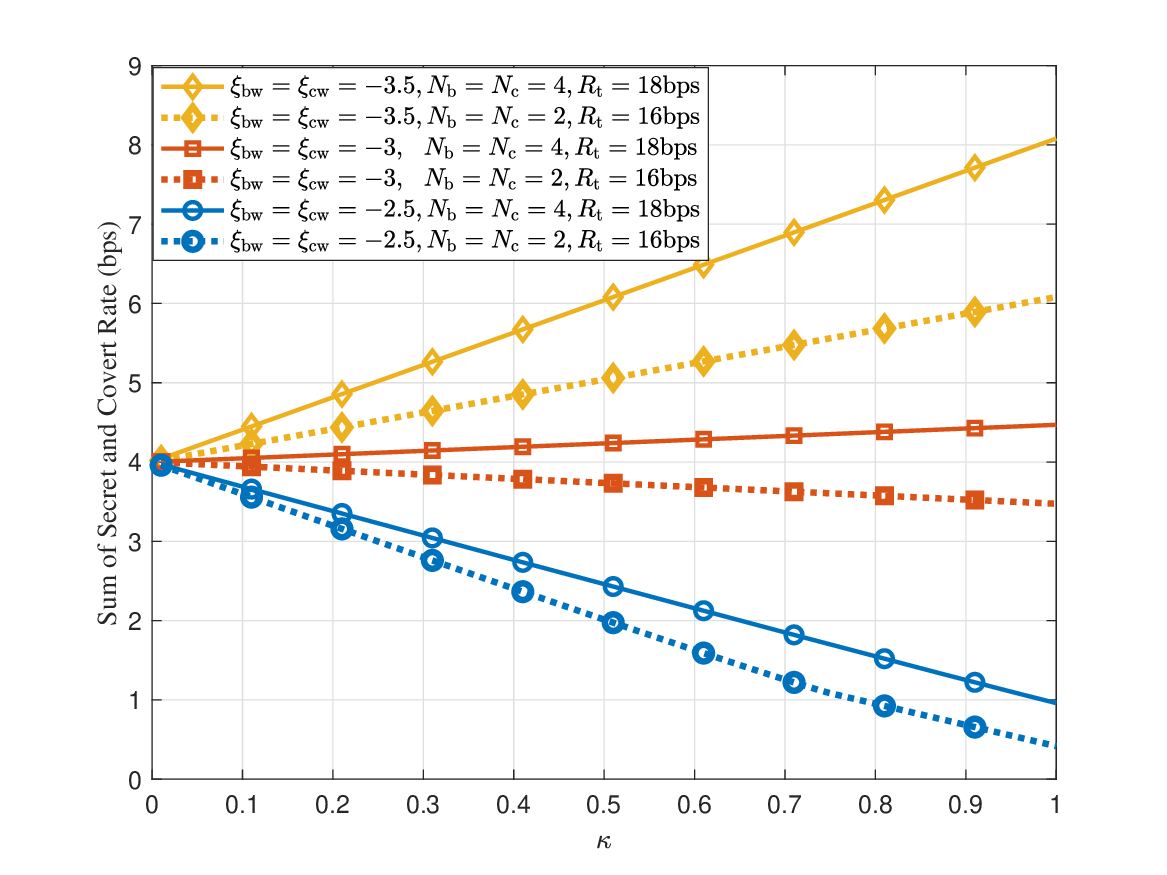}
\caption{Numerical plots of the sum of secret and covert rates versus the weight $\kappa$ with different Willie's path loss at the UAV's original position, where $Q_\rmb^\mathrm{max}=30$ dBm, $Q_\rmc^\mathrm{max}= 0$ dBm, $R_\rmc=4$ bps,  $\xi_\mathrm{ba}=\xi_\mathrm{ca}=-2$, $\xi_\mathrm{bw}=\xi_\mathrm{cw}=-3$, and $\sigma _\rma^2=\sigma _\rmw^2=-90$ dBm.}
\label{fig:sumratevspi}
\end{figure}
\begin{figure}[t]
\centering
\includegraphics[width=1\linewidth]{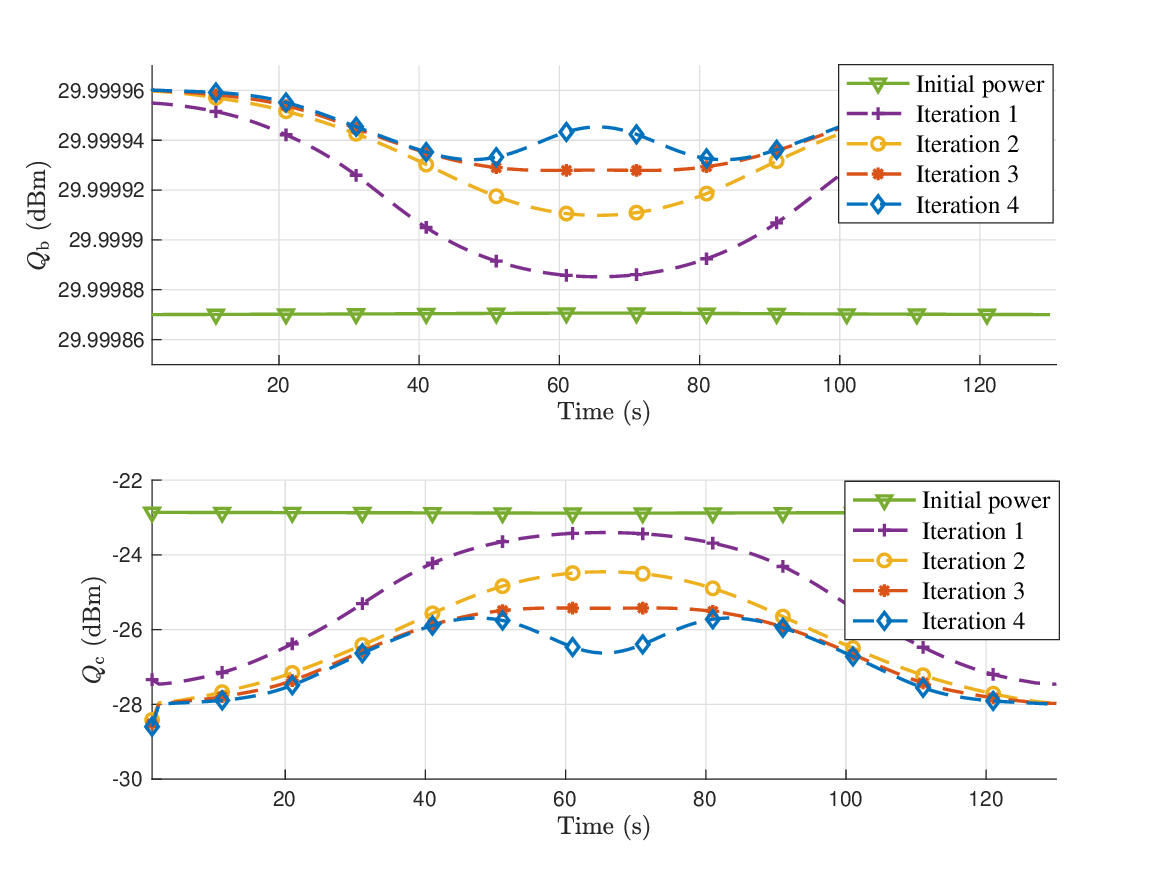}
\caption{Numerical plots of optimized transmit powers of Bob and Carlo during the UAV trajectory, where $Q_\rmb^\mathrm{max}=30$ dBm, $Q_\rmc^\mathrm{max}=0$ dBm, $R_\rmt=16$ bps, $R_\rmc=4$ bps, $\kappa=0.5$, $\xi_\mathrm{ba}=\xi_\mathrm{ca}=-2$, $\xi_\mathrm{bw}=\xi_\mathrm{cw}=-3$, and $\sigma _\rma^2=\sigma _\rmw^2=-90$ dBm.}
\label{fig:opt_power} 
\end{figure}
\begin{figure}[t]
\centering
\includegraphics[width=1\linewidth]{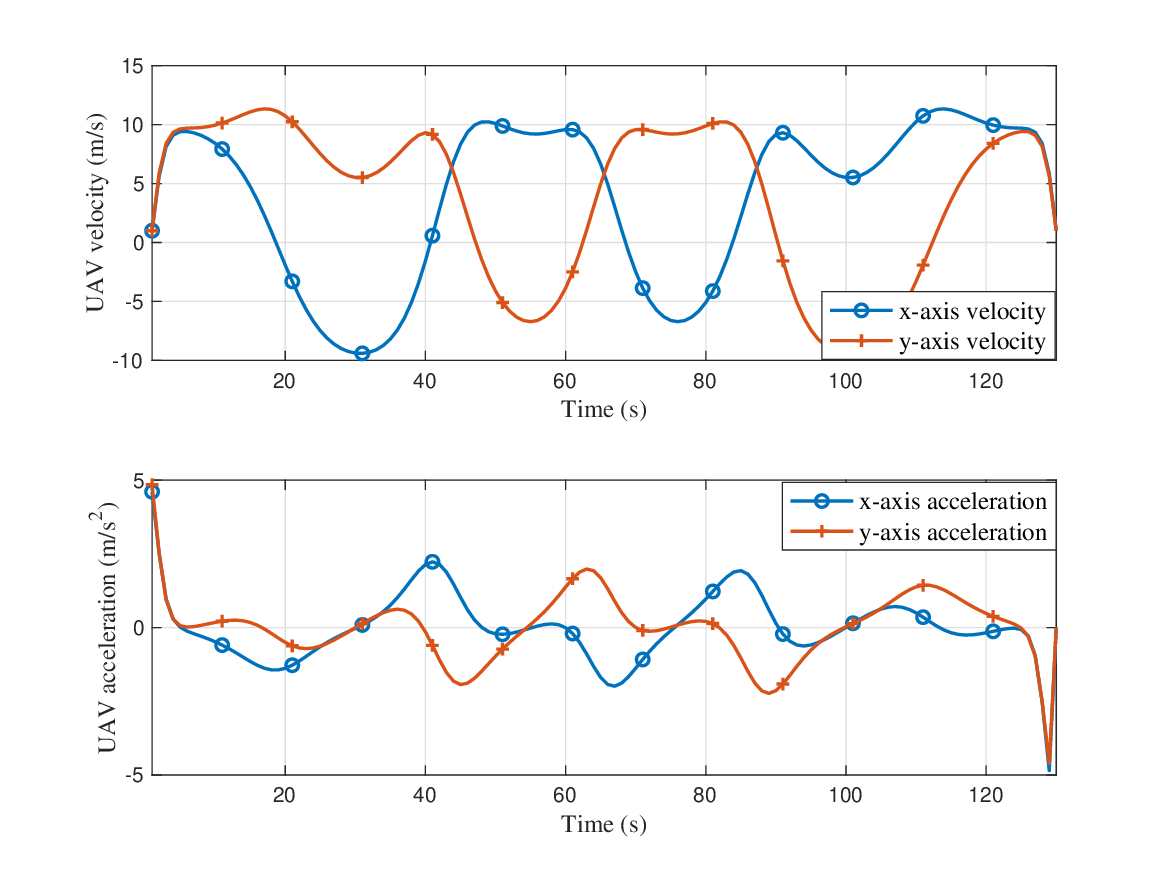}
\caption{Numerical plots of optimized velocity and acceleration during the UAV trajectory, where $Q_\rmb^\mathrm{max}=30$ dBm, $Q_\rmc^\mathrm{max}=0$ dBm, $R_\rmt=16$ bps, $R_\rmc=4$ bps, $\kappa=0.5$, $\xi_\mathrm{ba}=\xi_\mathrm{ca}=-2$, $\xi_\mathrm{bw}=\xi_\mathrm{cw}=-3$, and $\sigma _\rma^2=\sigma _\rmw^2=-90$ dBm, $v_\mathrm{min}=1$ m/s, $v_\mathrm{max}=20$ m/s, $a_\mathrm{max}=10$ m/s.}
\label{fig:opt_speed} 
\end{figure}
\begin{figure}[t]
\centering
\includegraphics[width=1\linewidth]{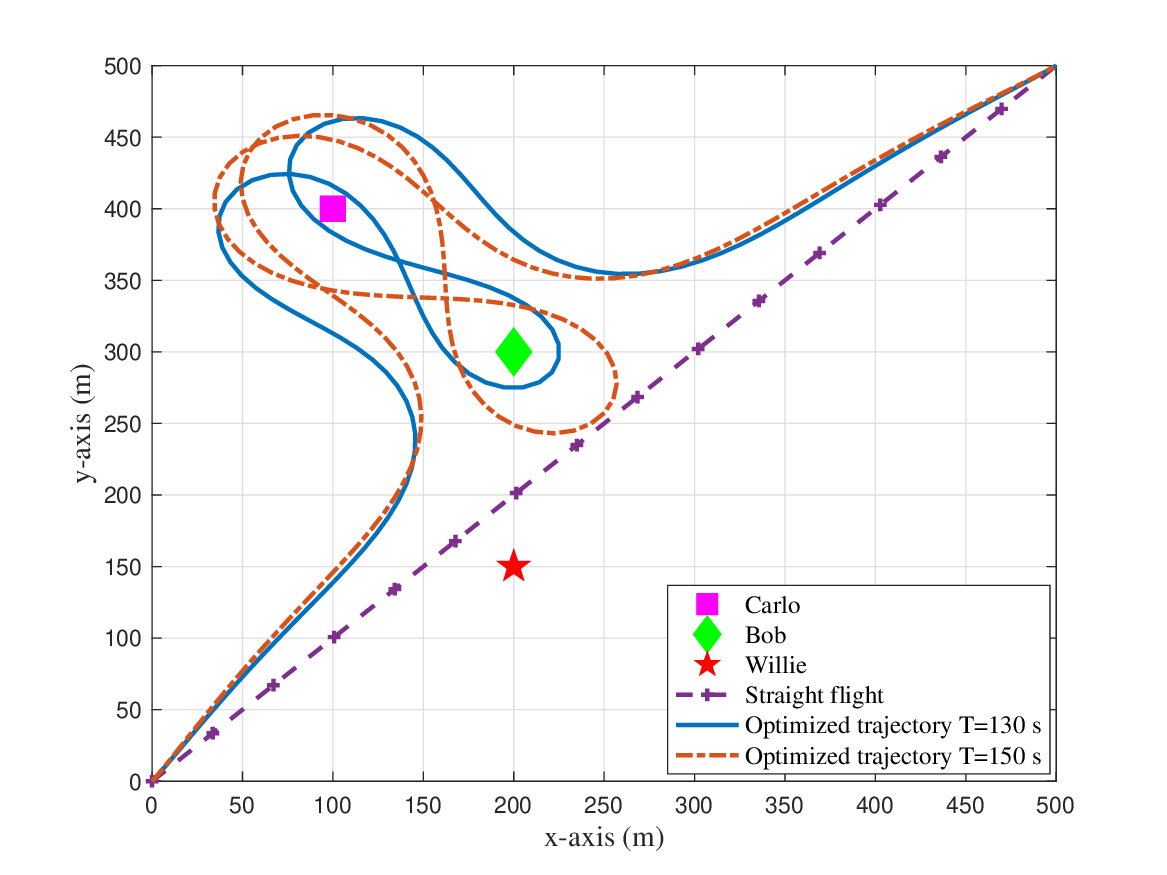}
\caption{Numerical plots of optimized UAV trajectories under different flight periods, where $Q_\rmb^\mathrm{max}=30$ dBm, $Q_\rmc^\mathrm{max}= 0$ dBm, $R_\rmt=16$ bps, $R_\rmc=4$ bps, $\xi_\mathrm{ba}=\xi_\mathrm{ca}=-2$, $\xi_\mathrm{bw}=\xi_\mathrm{cw}=-3$, and $\sigma _\rma^2=\sigma _\rmw^2=-90$ dBm, $v_\mathrm{min}=1$ m/s, $v_\mathrm{max} =20$ m/s, $a_\mathrm{max} =10$ m/s.}
\label{fig:UAV_traj} 
\end{figure}

{In Fig. \ref{fig:sumratevspi}, we plot the sum of secret and covert rates versus the weight $\kappa$. We keep the air-ground link's path loss constant while gradually increasing the path loss of Willie's links to the legitimate users Bob and Carlo. From Fig. \ref{fig:sumratevspi}, we can observe that Willie's channel quality significantly impacts the sum of the secret and covert rates. Specifically, when Willie's path loss is close to that of Bob, the covert rate dominates the sum of the security rates. On the other hand, the secret rate dominates the security rates when Willie's path loss deteriorates. This observation indicates that either the secret rate or the covert rate can dominate the sum of the security rates. Furthermore, since increasing the number of antennas increases the SCP and decreases the SOP as shown in Figs. \ref{fig:scp} and \ref{fig:sop}, we increase the number of Bob's antennas and Bob's target transmission rate. Comparing the solid and dotted curves of the same color in Fig. \ref{fig:sumratevspi}, we find that the sum of secret and covert rates increases as Bob's target transmission rate increases, particularly when Willie’s channel deteriorates, although Willie might obtain a larger diversity gain in this case. This phenomenon further highlights the benefits of increasing the number of transmit antennas to improve the system security performance.}

In Figs. \ref{fig:opt_power}, \ref{fig:opt_speed}, and \ref{fig:UAV_traj}, we plot the optimized transmit powers of Bob and Carlo in each iteration, the optimized UAV's velocity and acceleration, and the optimized UAV's trajectory, respectively. In Fig. \ref{fig:opt_power}, based on the perfect covert constrict in Eq. \eqref{45}, the transmit powers of Bob and Carlo is optimized during the UAV flight period, respectively. In Fig. \ref{fig:opt_speed}, we respectively plot the optimized UAV's velocity and the acceleration, which corresponds to the UAV's trajectory in Fig. \ref{fig:UAV_traj} when $T=130$ s. In Fig. \ref{fig:UAV_traj}, we plot the UAV's trajectories in different flight periods. {Compared to straight-line flight, our proposed algorithm effectively optimizes the UAV's trajectory according to different flight periods. In summary, the secret and covert beamformers, along with UAV trajectory-related parameters, can be jointly optimized using the proposed SCA-BCD algorithm to yield superior performance.} 

\section{Conclusion} \label{sec:conclusion}
We proposed a multi-user collaborative secret and covert uplink transmission scheme using NOMA and analyzed the performance of our scheme via optimization problems that maximized the sum of secret and covert rates under constraints on security and reliability. Our results provided feasible solutions for the designs of the beamforming method of ground users and UAV's trajectory and revealed the influence of channel parameters and design parameters. There are several avenues for future studies. Firstly, we assumed quasi-static fading channels when a fixed-wing UAV is used. In contrast, when a rotation-wing UAV is used, the block fading channel is practical and one can generalize our results to this case to account for both uplink and downlink multi-user UAV communications~\cite{hua20203d}. Secondly, we assumed perfect CSI knowledge at the UAV receiver and a single adversary. However, practical communication systems usually involve multiple adversaries and CSI estimation can be imperfect. Thus, one can generalize our results to the case of multiple adversaries with imperfect CSI. To do so, one might need to combine ideas in covert transmission over imperfect CSI~\cite{ma2021robust}, multiple adversaries~\cite{chen2021uav_tvt}, the uncertainty of the eavesdropper's location~\cite{jiang2021resource}, and a multiple antenna adversary~\cite{shahzad2019covert}. {Finally, we assumed that both Alice and Willie have a single antenna. However, MIMO communication is common in practice \cite{XJP_Model}. Thus, it is of value to explore the effect of MIMO  configurations for both Alice and Willie~\cite{bai2023covert} on the performance of joint secret and covert communications in UAV systems.}

\bibliographystyle{IEEEtran}
\bibliography{IEEE_UAVcovert}   
\end{document}

%% file: preamble.tex
\usepackage{soul}
\usepackage[mathscr]{eucal}
\usepackage[cmex10]{amsmath}
 \usepackage{epsfig,epsf,psfrag}
\usepackage{amssymb,amsmath,amsthm,amsfonts,latexsym}
\usepackage{amsmath,graphicx,bm,xcolor,url,overpic}
\usepackage[caption=false]{subfig} 
\usepackage{fixltx2e}%ordering of single and double column floats
\usepackage{array}%array and tabular environments
\usepackage{verbatim}
\usepackage{bm}
\usepackage{algorithmic}
\usepackage{algorithm}
\usepackage{verbatim}
\usepackage{textcomp}
\usepackage{mathrsfs}
\usepackage{epstopdf}

\newcommand{\openone}{\leavevmode\hbox{\small1\normalsize\kern-.33em1}}

\catcode`~=11 \def\UrlSpecials{\do\~{\kern -.15em\lower .7ex\hbox{~}\kern .04em}} \catcode`~=13 

\allowdisplaybreaks[4]

\newcommand{\nn}{\nonumber}

\newcommand{\calX}{\mathcal{X}}

\newcommand{\ba}{\mathbf{a}}
\newcommand{\bA}{\mathbf{A}}

\newcommand{\bg}{\mathbf{g}}

\newcommand{\bh}{\mathbf{h}}

\newcommand{\bI}{\mathbf{I}}

\newcommand{\bL}{\mathbf{L}}

\newcommand{\bn}{\mathbf{n}}

\newcommand{\bv}{\mathbf{v}}

\newcommand{\bw}{\mathbf{w}}
\newcommand{\bW}{\mathbf{W}}
\newcommand{\bx}{\mathbf{x}}

\newcommand{\by}{\mathbf{y}}

\newcommand{\bz}{\mathbf{z}}

\newcommand{\rma}{\mathrm{a}}

\newcommand{\rmb}{\mathrm{b}}

\newcommand{\rmc}{\mathrm{c}}

\newcommand{\rmd}{\mathrm{d}}

\newcommand{\rme}{\mathrm{e}}

\newcommand{\rmg}{\mathrm{g}}

\newcommand{\rmH}{\mathrm{H}}

\newcommand{\rmm}{\mathrm{m}}

\newcommand{\rmP}{\mathrm{P}}

\newcommand{\rms}{\mathrm{s}}

\newcommand{\rmt}{\mathrm{t}}

\newcommand{\rmw}{\mathrm{w}}

\newcommand{\bbN}{\mathbb{N}}

\DeclareMathAlphabet{\mathbsf}{OT1}{cmss}{bx}{n}
\DeclareMathAlphabet{\mathssf}{OT1}{cmss}{m}{sl}% slanted sans serif

\DeclareSymbolFont{bsfletters}{OT1}{cmss}{bx}{n}  
\DeclareSymbolFont{ssfletters}{OT1}{cmss}{m}{n}
\DeclareMathSymbol{\bsfGamma}{0}{bsfletters}{'000}
\DeclareMathSymbol{\ssfGamma}{0}{ssfletters}{'000}
\DeclareMathSymbol{\bsfDelta}{0}{bsfletters}{'001}
\DeclareMathSymbol{\ssfDelta}{0}{ssfletters}{'001}
\DeclareMathSymbol{\bsfTheta}{0}{bsfletters}{'002}
\DeclareMathSymbol{\ssfTheta}{0}{ssfletters}{'002}
\DeclareMathSymbol{\bsfLambda}{0}{bsfletters}{'003}
\DeclareMathSymbol{\ssfLambda}{0}{ssfletters}{'003}
\DeclareMathSymbol{\bsfXi}{0}{bsfletters}{'004}
\DeclareMathSymbol{\ssfXi}{0}{ssfletters}{'004}
\DeclareMathSymbol{\bsfPi}{0}{bsfletters}{'005}
\DeclareMathSymbol{\ssfPi}{0}{ssfletters}{'005}
\DeclareMathSymbol{\bsfSigma}{0}{bsfletters}{'006}
\DeclareMathSymbol{\ssfSigma}{0}{ssfletters}{'006}
\DeclareMathSymbol{\bsfUpsilon}{0}{bsfletters}{'007}
\DeclareMathSymbol{\ssfUpsilon}{0}{ssfletters}{'007}
\DeclareMathSymbol{\bsfPhi}{0}{bsfletters}{'010}
\DeclareMathSymbol{\ssfPhi}{0}{ssfletters}{'010}
\DeclareMathSymbol{\bsfPsi}{0}{bsfletters}{'011}
\DeclareMathSymbol{\ssfPsi}{0}{ssfletters}{'011}
\DeclareMathSymbol{\bsfOmega}{0}{bsfletters}{'012}
\DeclareMathSymbol{\ssfOmega}{0}{ssfletters}{'012}

\newcommand{\bmu}{\bm{\mu}}

\newtheorem{theorem}{Theorem} 
\newtheorem{lemma}[theorem]{Lemma}